\documentclass[12pt,reqno,twoside,fleqn,a4paper]{amsart}%
\usepackage[mathcal]{euler}
\usepackage{palatino}
\usepackage{amsmath,amssymb,amsthm,graphics}%
\usepackage{tikz-cd}
\usetikzlibrary{matrix,arrows,decorations.pathmorphing}

\numberwithin{equation}{section}
\makeatletter
\renewcommand{\section}{\@startsection {section}{1}{\z@}%
                                   {-3.5ex \@plus -1ex \@minus -.2ex}%
                                   {.5\linespacing}%
                                   {\normalfont\scshape\centering}}
\makeatother

\newtheorem{thm}{Theorem}[section]
\newtheorem{lem}[thm]{Lemma}
\newtheorem{cor}[thm]{Corollary}

\theoremstyle{definition}
\newtheorem{definition}{Definition}[section]
\newtheorem{axiom}{Axiom}

\theoremstyle{remark}
\newtheorem*{example}{Example}
\newtheorem*{remark}{Remark}

\def\beq#1\eeq{\begin{equation}#1\end{equation}}
\makeatletter
\@ifpackageloaded{euler}{
 \newcommand{\onto}{\to\mkern-14mu\to}
 \def\hbar{{\mathchar'26\mkern-6.5muh}}
 
}{
 \DeclareMathSymbol{\onto}{\mathrel}{AMSa}{"10}
 \renewcommand{\hbar}{{\mathchar'26\mkern-9muh}}
 
}
\makeatother
\newcommand{\C}{\mathcal{C}}
\newcommand{\Ci}{\C^\infty}
\newcommand{\co}{\mathbb{C}}
\newcommand{\cs}{\mbox{\upshape C}\ensuremath{{}^*}}

\newcommand{\R}{\mathbb{R}}

\DeclareMathOperator{\id}{id}

\newcommand{\inner}{\mathbin{\raise1.5pt\hbox{$\lrcorner$}}}
\newcommand{\K}{\mathcal{K}}

\newcommand{\D}{\mathcal{D}}

\renewcommand{\Im}{\mathop{\mathrm{Im}}}

\newcommand{\m}{\mathsf{m}}

\newcommand{\X}{\mathcal{X}}

\DeclareMathOperator{\Exp}{Exp}

\newcommand{\U}{\mathrm{U}}

\DeclareMathOperator{\Ad}{Ad}

\newcommand{\norm}[1]{\lVert#1\rVert}
\newcommand{\naturalto}{\mathrel{\dot\to}}
\newcommand{\modification}{\mathrel{\ddot\to}}
\newcommand{\Vect}{\mathsf{Vec}}
\newcommand{\Alg}{\mathsf{Alg}}
\newcommand{\Xc}{\mathfrak{X}}
\DeclareMathOperator{\Hom}{Hom}

\newcommand{\comp}{\mathbin{\bar\circ}}
\newcommand{\Loc}{\mathsf{Loc}}
\newcommand{\A}{\mathfrak{A}}
\newcommand{\B}{\mathfrak{B}}
\newcommand{\Cf}{\mathfrak{C}}
\DeclareMathOperator{\Supp}{Supp}
\newcommand{\Or}{\mathcal O}
\DeclareMathOperator{\tot}{tot}

\DeclareMathOperator{\Aut}{Aut}
\newcommand{\Zcenter}{\mathcal{Z}}
\newcommand{\sAlg}{\text{$*$-$\Alg$}}
\DeclareMathOperator{\Obj}{Obj}
\DeclareMathOperator{\Mor}{Mor}

\newcommand{\der}{\operatorname{\mathfrak{der}}}
\newcommand{\op}{\mathrm{op}}
\DeclareMathOperator{\SAut}{SAut}
\DeclareMathOperator{\SOut}{SOut}
\newcommand{\dH}{\delta^{\mathrm H}}
\newcommand{\dS}{\delta^{\mathrm S}}
\newcommand{\Int}{V}
\newcommand{\Cat}{\mathsf{Cat}}
\newcommand{\Dt}{\mathcal{D}}
\newcommand{\dt}{\cdot_{\scriptscriptstyle\mathcal{T}}}
\newcommand{\ET}{\Exp_{\mathcal{T}}}
\newcommand{\ess}{\mathcal{S}}
\newcommand{\Moller}{M{\o}ller}
\newcommand{\mm}{\tilde{R}}
\newcommand{\hAlg}{\Alg[[\hbar]]}
\newcommand{\hlAlg}{\Alg[[\hbar,\lambda]]}
\newcommand{\Vcal}{\mathcal{V}}
\newcommand{\Lcal}{\mathcal{L}}
\newcommand{\Algt}{\mathsf{AlgInn}}

\title[A Cohomological Perspective on Algebraic Quantum Field Theory]{A Cohomological Perspective\\ on Algebraic Quantum Field Theory}
\author{Eli Hawkins}

\begin{document}
\maketitle
\begin{center}
\vspace{-4ex}
\emph{\small Department of Mathematics}\\
\emph{\small The University of York, United Kingdom}\\
{\small eli.hawkins@york.ac.uk}\\
\end{center}

\begin{abstract}
Algebraic quantum field theory is considered from the perspective of the Hochschild cohomology bicomplex. This is a framework for studying deformations and symmetries. Deformation is a possible approach to the fundamental challenge of constructing interacting QFT models. Symmetry is the primary tool for understanding the structure and properties of a QFT model.
  
This perspective leads to a generalization of the algebraic quantum field theory framework, as well as a more general definition of symmetry. This means that some models may have symmetries that were not previously recognized or exploited.

To first order, a deformation of a QFT model is described by a Hochschild cohomology class. A deformation could, for example, correspond to adding an interaction term to a Lagrangian. The cohomology class for such an interaction is computed here. However, the result is more general and does not require the undeformed model to be constructed from a Lagrangian. This computation leads to a more concrete version of the construction of perturbative algebraic quantum field theory.
\end{abstract}

\tableofcontents

\section{Introduction}
At present, the best known description of fundamental physics is given by the Standard Model, which is an interacting quantum field theory (QFT) in 4 dimensions. Unfortunately, a mathematically consistent description of the Standard Model is not yet known. It is a fundamental problem of mathematical physics to construct mathematical models of interacting quantum field theories such as the Standard Model or whatever may supplant it. 

Algebraic Quantum Field Theory (AQFT) \cite{haa1996} does provide a framework for describing QFT, but thus far, interacting models have only been constructed in dimensions less than 4. More tools are needed for constructing and understanding quantum field theories.

A possible approach to constructing interacting QFTs is by deformation --- either deforming a free QFT into an interacting one or deforming an interacting classical field theory into a quantum one. This is analogous to deforming a commutative algebra into a noncommutative one, as is done in formal \cite{bffls,kon} or strict \cite{rie11} deformation quantization. 

Formal deformation quantization is part of the theory of algebraic deformations \cite{ger1964}, which is based upon Hochschild cohomology and the algebraic structure of the Hochschild complex. The purpose of this paper is to consider AQFT from the perspective of the relevant generalization of Hochschild cohomology. This is a necessary step toward a theory of deformation quantization of field theories and thus an approach to building interacting QFT models.

This perspective provides a unified framework for three seemingly disparate concepts: the symmetries of a QFT, the transition from classical to quantum field theory, and the transition from free to interacting QFT. It also leads to a more general definition of symmetry and a generalization of AQFT.

\subsection{Algebraic quantum field theory}
The fundamental difference between quantum field theory and other models of quantum physics is locality. Consistency with relativity means that only some observables can be measured in a given region, $\Or$, of spacetime. Observables regarding processes spacelike separated from $\Or$ cannot be measured in $\Or$. This is a manifestation of the principle that no signal can travel faster than light.

Any sum or product of observables that can be measured in $\Or$ can also be measured there, therefore the set of observables measurable in $\Or$ is an algebra, $\A(\Or)$. If $\Or_1\subset\Or_2$, then any observable that can be measured in $\Or_1$ can \emph{a fortiori} be measured in $\Or_2$, so $\A(\Or_1)\subseteq \A(\Or_2)$.

This correspondence between regions and algebras completely encodes the structure of a quantum field theory. This is the fundamental idea of Algebraic Quantum Field Theory \cite{haa1996}. 

This can be used to describe a QFT on Minkowski spacetime or on a curved spacetime. Locally Covariant Quantum Field Theory extends this to describe a QFT on all possible spacetimes. The key insight is that a region of a spacetime is itself a spacetime --- thus regions and spacetimes can be treated on the same footing. 
\begin{definition}
\cite{bfv2003}
$\Loc$ is the category in which:
\begin{itemize}
\item
the objects are all oriented, time-oriented, globally hyperbolic, Lorentzian manifolds (of some fixed dimension, $n$);
\item
the morphisms are the smooth maps that are  isometric, injective, oriented, and preserve the causal relation.
\end{itemize}
\end{definition}
\begin{definition}
\cite{bfv2003} A \emph{Locally covariant quantum field theory} (LCQFT)  is  a covariant functor  $\A:\Loc\to\Alg$ (or to some other category of algebras). 
\end{definition}
This is usually required to satisfy further axioms. 
\begin{axiom}[Einstein Causality]
If $\iota_1:\Or_1\to M$ and $\iota_2:\Or_2\to M$ and the images of $\iota_1$ and $\iota_2$ are spacelike separated, then the images of $\A[\iota_1]$ and $\A[\iota_2]$ commute.
\end{axiom}
\begin{axiom}[Time Slice]
If $\phi : M\to N$ and $\Im \phi$ contains a Cauchy surface of $N$, then $\A[\phi]: \A(M)\to\A(N)$ is an isomorphism.
\end{axiom}
\begin{axiom}[Isotony]
For any $\phi$, $\A[\phi]$ is injective.
\end{axiom}
In fact, only Einstein causality will be needed in this paper.

The category $\Loc$ is monoidal under the operation of disjoint union of spacetimes. Einstein causality is almost equivalent to requiring $\A$ to be a monoidal functor (see \cite{bfir2014}).

If $\Xc\subset\Loc$ is a small category whose inclusion is an equivalence of categories, then an LCQFT can equivalently be described as a functor $\A:\Xc\to\Alg$. For example, $\Xc$ could be the subcategory of spacetimes whose underlying manifolds are submanifolds of $\R^{2n+1}$. The results of \cite{m-s} imply that $\Loc$ is equivalent to the small subcategory of globally hyperbolic submanifolds of Minkowski spacetime (of sufficiently large dimension).

A QFT on a fixed manifold, $M\in\Obj(\Loc)$, can also be encoded as a functor. If $\Xc\subset\Loc$ is the subcategory of spacetimes that happen to be open subsets of $M$, then a QFT on $M$ can be encoded as a functor $\A:\Xc\to\Alg$. Einstein causality and the time slice axiom are perfectly meaningful conditions on such a functor.
Note that this encodes the action of any oriented, time-oriented isometries of $M$. In particular, if $M$ is Minkowski spacetime, then this functor encodes the action of the Poincar\'e group.

A similar approach can be taken to conformal field theory.  
A conformal net is a functor from a category of open intervals in $S^1$  to von~Neumann algebras (satisfying further axioms). See, e.g., \cite{gl1996}.

A cruder description of quantum physics on a fixed spacetime, ignoring locality, can also be described in this way. Given $M\in\Obj(\Loc)$, the full subcategory of $\Loc$ with the single object $M$ is the group of oriented, time-oriented isometries of $M$. A functor from this group (as a category) to $\Alg$ encodes the algebra of observables on $M$ and the action of this group on that algebra.

For most of this paper, I will talk about an arbitrary small category, $\Xc$. I have in mind any of the examples above. In the later sections, this will be limited to a subcategory $\Xc\subset\Loc$ for which Einstein causality is a meaningful condition.

\subsection{Hochschild cohomology}
The continuous functions on a topological space and the smooth functions on a manifold form commutative algebras. Many geometrical constructions can be expressed algebraically in terms of these commutative algebras and extend easily to noncommutative algebras. It is often useful to view a noncommutative algebra as if it comes from a topological space and to apply geometrical ideas and intuition. This is the fundamental idea of Noncommutative Geometry.

For example, let $M$ be a compact, smooth manifold, and $\X^\bullet(M)$ the space of smooth, antisymmetric multivector fields. This is a \emph{Gerstenhaber algebra} with both a graded commutative, associative product (the exterior product) and a graded Lie bracket (the Schouten-Nijenhuis bracket). The Hochschild cohomology $H^\bullet(A,A)$ of the commutative algebra $A = \Ci(M)$ is naturally identified with $\X^\bullet(M)$ as a graded vector space. 

Moreover, Gerstenhaber constructed a graded Lie bracket and an associative product on the Hochschild complex $C^\bullet(A,A)$ of any algebra, which give the cohomology the structure of a Gerstenhaber algebra (hence the name) and for $\Ci(M)$ this is the natural structure mentioned in the last paragraph. This means that Hochschild cohomology should be thought of as a noncommutative generalization of the Gerstenhaber algebra of multivector fields. This --- and the detailed structure on the complex --- play a central role in the theory of formal deformation quantization.

\subsection{Algebraic quantum field theory as noncommutative geometry}
\begin{definition}
A \emph{diagram of algebras} is a covariant functor from a small category to $\Alg$.
\end{definition}
As described above, an AQFT can be expressed as a functor $\A:\Xc\to\Alg$, where $\Xc\subset\Loc$ is a small subcategory, thus an AQFT is a diagram of algebras.
\begin{remark}
In \cite{gs1988a,gs1988b} a diagram is defined as a presheaf (contravariant functor) but an AQFT is covariant, and the difference is just a matter of replacing $\Xc$ with its opposite category.
\end{remark}
Let $\boldsymbol 1$ be the category with one object and one morphism. A single algebra is trivially equivalent to a functor $\boldsymbol1\to\Alg$. Thus:
\begin{itemize}
\item
An AQFT is in particular a diagram of algebras.
\item
A diagram of algebras is a generalization of an algebra.
\item
An algebra is a generalization of an algebra of functions on a space.
\end{itemize}
In this way, QFT is a generalization of geometry. This is the perspective that I will pursue here.

\subsection{Notation and terminology}
\label{Notation}
$\Vect$ and $\Alg$ will denote the categories of vector spaces and algebras over the field of complex numbers, $\co$. $\sAlg$ will denote the category of $*$-algebras.

For any $M\in\Obj(\Loc)$, let $\Dt(M) := \Ci_c(M,\R)$ be the space of smooth, compactly supported functions (test functions). Because any $\Loc$ morphism $\phi:M\to N$ is injective and open, there is a push-forward map $\phi_*:\Dt(M)\to\Dt(N)$; for $f\in\Dt(M)$, $\phi_*f$ is defined by the conditions that $\phi^*\phi_*f=f$ and $\Supp(\phi_*f) \subseteq \Im\phi$. If we define $\Dt[\phi]:=\phi_*$, then $\Dt:\Loc\to\Vect$ is a covariant functor.

Given two points $x,y\in M\in\Obj(\Loc)$, denote \cite{bar,h-e}:
\begin{itemize}
\item
$x\leq y$ if there exists a future-directed causal curve from $x$ to $y$.
\item
$x<y$ if $x\leq y$ and $x\neq y$.
\item
$x\sim y$ if neither $x<y$ nor $x>y$.
\item
$x \lesssim y$ if $x<y$ or $x\sim y$.
\end{itemize}
These relations extend to sets. For example, if $\Or_1,\Or_2\subset M$, then write
\[
\Or_1\lesssim \Or_2 \Longleftrightarrow \forall x\in \Or_1,y\in\Or_2: x\lesssim y \ .
\]

If $\Or\subseteq M \in\Obj(\Loc)$, then denote
\[
J^+_M(\Or) := \{x\in M \mid \exists y\in \Or : x\geq y\}
\]
and
\[
J^-_M(\Or) := \{x\in M \mid \exists y\in \Or : x\leq y\}
\]
(or $J^\pm(\Or)$, if there is no ambiguity). Further, let $J_M(\Or) := J^+_M(\Or)\cup J^-_M(\Or)$.
The \emph{causal complement} of $\Or$ is
$\Or' :=  M\smallsetminus J_M(\Or)$
and
$\Or$ is \emph{causally complete} if $\Or = \Or''$.

A subset $\Or\subset M\in\Obj(\Loc)$ is \emph{future/past compact} if for any $x\in M$, $J^\pm(x)\cap \Or$ is compact.

A function $V$ (not necessarily linear) from $\Dt(M)$ to a vector space is \emph{additive} if for $f,g,h\in\Dt(M)$,
\[
\Supp f \cap \Supp h = \emptyset \implies V(f+g+h) = V(f+g)-V(g)+V(g+h) \ .
\]
This is really a locality condition, but note that a linear map is automatically additive.

The bracket notation $[\;\cdot\;,\;\cdot\;]$ will be used for both the Gerstenhaber bracket and the commutator in an associative algebra. I hope that this will be clear in context.

\subsection{Outline}
In Section \ref{Hochschild}, I review the definitions of Hochschild cohomology and the Gerstenhaber algebra structure for a single algebra and for a diagram of algebras. 

In Section~\ref{Significance}, I discuss the relationship of ``asimplicial'' Hochschild cohomology to the deformations and automorphisms of a diagram of algebras. Seeking a similar interpretation of full Hochschild cohomology leads me to define skew diagrams of algebras and their morphisms. This gives the first main results: a generalization of AQFT and a more general definition of global symmetries of an AQFT. The category of skew diagrams is shown to be a 2-category of functors between 2-categories.

The first main calculation is in Section~\ref{Interaction}, where I compute the characteristic class in Hochschild cohomology of an interaction term for an AQFT. This involves defining a smoothed-out analogue of Cauchy surfaces. 

In Section~\ref{Perturbative}, I discuss the construction of perturbative AQFT by the algebraic adiabatic limit. The next main result is an alternative,  more concrete construction; this is motivated by my computation of the characteristic class. This construction leads to the last main calculation --- a direct proof that the characteristic class satisfies the appropriate Maurer-Cartan equation.

\section{Hochschild Cohomology}
\label{Hochschild}
The definition of Hochschild cohomology $H^\bullet(A,A)$ for an algebra extends to diagrams of algebras\footnote{The functorial Hochschild cohomology $HH^\bullet(A)=H^\bullet(A,A^*)$ does not extend to diagrams, because $A\mapsto A^*$ is not a covariant functor.}. This is referred to as Yoneda cohomology by Gerstenhaber and Schack in \cite{gs1988a} because they were working with algebras over a ground ring that was not necessarily a field; that degree of generality is irrelevant here. 

Hochschild cohomology of a diagram of algebras is still a Gerstenhaber algebra. This cohomology governs deformations of diagrams, just as it does for a single algebra. 

This does not perfectly characterize deformations of LCQFTs, because a LCQFT might be deformed to a diagram of algebras that violates Einstein causality. Nevertheless, this does describe a lot of the relevant structure, and an infinitesimal deformation will have a characteristic class in Hochschild cohomology.

The category of algebras in AQFT is most often taken to be \cs-algebras or von~Neumann algebras. These are not well suited for studying infinitesimal deformations. To construct multivector fields via Hochschild cohomology, we use not the \cs-algebra of continuous functions but the dense subalgebra of smooth functions. This suggests that studying infinitesimal deformations of an LCQFT may require identifying analogous dense subalgebras.

The main explicit calculation here will be in the setting of perturbative LCQFT, which does not use \cs-algebras.

Let's begin by recalling the definition and properties of Hochschild cohomology for an algebra.
\subsection{A single algebra}
\cite{hoc1945}
Let $A$ be an associative algebra over $\co$, and $B$ a bimodule of $A$.
\subsubsection{The complex}
\begin{definition}
$C^q(A,B) := \Hom_\co(A^{\otimes q},B)$ is the space of $q$-multilinear maps.
\end{definition}
\begin{definition}
The maps $\dH_i:C^q(A,B)\to C^{q+1}(A,B)$ are defined by, for any $\Gamma\in C^q(A,B)$ and $a_1,\dots,a_{q+1}\in A$, 
\[
(\dH_0\Gamma)(a_1,\dots,a_{q+1}) := a_1 \Gamma(a_2,\dots,a_{q+1})\ ,
\]
\[
(\dH_i\Gamma)(a_1,\dots,a_{q+1}) := \Gamma(a_1,\dots,a_ia_{i+1},\dots,a_{q+1})
\]
for $1\leq i\leq q$, and
\[
(\dH_{q+1}\Gamma)(a_1,\dots,a_{q+1}) := \Gamma(a_1,\dots,a_q)a_{q+1}\ .
\]
The Hochschild coboundary $\delta:C^q(A,B)\to C^{q+1}(A,B)$ is 
\[
\delta := \sum_{i=0}^{q+1} (-1)^i \dH_i\ .
\]
\end{definition}
\begin{definition}
The Hochschild cohomology $H^\bullet(A,B)$ is the cohomology of $C^\bullet(A,B)$ with the coboundary $\delta$.
\end{definition}

\subsubsection{The Gerstenhaber bracket}
\cite{ger1963}
Now consider the case that $B=A$.

Let $\Gamma \in C^q(A,A)$ and $\Delta \in C^{q'}(A,A)$.
\begin{definition}
For $1\leq i \leq q$, the \emph{partial composition} $\Gamma\circ_i\Delta \in C^{q+q'-1}(A,A)$ is defined by
\[
(\Gamma\circ_i\Delta)(a_1,\dots,a_{q+q'-1}) := \Gamma(a_1,\dots,a_{i-1},\Delta(a_i,\dots,a_{i+q-1}),a_{i+q},\dots,a_{q+q'-1})
\]
for $a_1,\dots,a_{q+q'-1}\in A$.
\end{definition}
\begin{definition}
From this, define 
\[
\Gamma\circ\Delta := \sum_{i=1}^q (-1)^{(q-i)(q'-1)} \Gamma\circ_i\Delta
\]
and the \emph{Gerstenhaber bracket}
\[
[\Gamma,\Delta] = \Gamma\circ\Delta - (-1)^{(q-1)(q'-1)} \Delta\circ\Gamma \in C^{q+q'-1}(A,A)\ .
\]
\end{definition}
This bracket is a graded Lie bracket of degree $-1$. Equivalently, $C^\bullet(A,A)$ with this bracket is a graded Lie algebra with the shifted grading in which $\Gamma\in\C^q(A,A)$ has degree $q-1$.

\begin{definition}
$\m\in C^2(A,A)$ is the multiplication map, i.e., $\m(a,b) := ab$.
\end{definition}
Note that $\delta\Gamma = [\m,\Gamma]$. From this, it is a simple exercise to deduce that $C^\bullet(A,A)$ is a differential graded Lie algebra. The defining property,
\[
\delta[\Gamma,\Delta] = [\delta\Gamma,\Delta] + (-1)^{q-1}[\Gamma,\delta\Delta]
\]
follows from the Jacobi identity. This implies that the Gerstenhaber bracket induces a well defined graded Lie bracket on the Hochschild cohomology $H^\bullet(A,A)$.

There is also an associative product.
\begin{definition}
$\Gamma\smile \Delta \in C^{q+q'}(A,A)$ is defined by
\[
(\Gamma\smile \Delta)(a_1,\dots,a_{q+q'}) := \Gamma(a_1,\dots,a_q)\Delta(a_{q+1},\dots,a_{q+q'})
\]
for $a_1,\dots,a_{q+q'}\in A$.
\end{definition}
This is obviously associative but is not commutative. Less obviously, this descends to an associative product on cohomology, where:
\begin{itemize}
\item
The product is commutative.
\item
The bracket is a derivation of the product (in each argument).
\end{itemize}

\subsubsection{Significance of the Gerstenhaber bracket}
\cite{ger1964}
Note that $\m\circ \m = \m\circ_2\m-\m\circ_1\m$, so
\[
(\m\circ\m)(a,b,c) = a(bc)-(ab)c
\]
and the equation 
\beq
\label{associativity}
0 = \m\circ\m = \tfrac12[\m,\m]
\eeq 
is precisely equivalent to the associativity of $\m$.

Imagine that $\m$ is part of a smooth, 1-parameter family of associative products. Differentiating eq.~\eqref{associativity} once gives
\[
0 = \m\circ\dot\m+\dot\m\circ\m = [\m,\dot\m] = \delta\m\ ,
\]
so an infinitesimal deformation of an associative product is a 2-cocycle. 
Differentiating again gives
\[
0 = [\dot\m,\dot\m]+[\m,\ddot\m] \implies [\dot\m,\dot\m] = -\delta\ddot\m\ ,
\]
so the Gerstenhaber bracket of $\dot\m$ with itself is exact. This and eq.~\eqref{associativity} are examples of Maurer-Cartan equations.

A deformation is trivial if $A$ with the deformed product is isomorphic to $A$ with the undeformed product. If $\alpha\in C^1(A,A)$ is such an isomorphism, then the deformed product of $a$ and $b$ is 
\[
\m(a,b)=\alpha^{-1}\left(\alpha(a)\alpha(b)\right)\ .
\]
Suppose that there is a 1-parameter family of such isomorphisms, starting from the identity. Differentiating this expression and then setting $\alpha=\id$ gives
\[
\dot\m(a,b) = \dot\alpha(a)b-\dot\alpha(ab) + a\dot\alpha(b)\ ,
\]
so $\dot\m = \delta\dot\alpha$. In other words, trivial infinitesimal deformations correspond to exact cocycles. This means that the Hochschild cohomology class of an infinitesimal deformation describes it modulo trivial deformations. 

Similarly, if there is a 1-parameter family of automorphisms, starting from the identity, then $0 = \dot\m = \delta\dot\alpha$. (This means precisely that $\dot\alpha$ is a derivation.) So, an infinitesimal automorphism is a 1-cocycle.

If $A$ is unital, an invertible element $b\in C^0(A,A) = A$ determines an inner automorphism, $\alpha(a) = b^{-1}ab$. Suppose that $b$ is part of a 1-parameter family, starting from the unit. Differentiating gives
\[
\dot\alpha(a)= a\dot b - \dot b a \implies \dot\alpha = \delta \dot b\ ,
\]
so infinitesimal inner automorphisms are exact 1-cocycles. This means that $H^1(A,A)$ describes infinitesimal automorphisms modulo inner ones.

Finally, the equation $0=\delta b$ is the condition that $b$ be central, so $H^0(A,A)=\Zcenter(A)$, the center of $A$.

Note that here there are various structures --- algebra elements, automorphisms, multiplication --- that are elements of $C^\bullet(A,A)$ in various degrees. These satisfy properties that are most naturally expressed as the vanishing of elements of $C^\bullet(A,A)$ in other degrees.

\subsection{A diagram of algebras}
\cite{mit,gs1988a,gs1988b} 
Let $\Xc$ be a small category and $\A:\Xc\to\Alg$  a covariant functor, i.e., a diagram of algebras over $\Xc$. Because I  mainly have in mind $\Xc\subset\Loc$, I will denote elements of $\Xc$ as $M$, $N$, \emph{et cetera}.

Such a functor consists of 3 types of information: Every object determines a vector space; every object also determines an associative product on that vector space; and every morphism determines a homomorphism of algebras. A vector space cannot be deformed, but the other two structures can. This is in contrast to a single algebra, where there is only one deformable structure.

Let $\phi:M\to N$ and $\psi:N\to P$ be morphisms in $\Xc$.

For a manifold $M$, there is a bilinear multiplication map $\m[M]: \A(M)^{\otimes 2} \to \A(M), a\otimes b \mapsto ab$. For a morphism $\phi$, there is a linear map, $\A[\phi]:\A(M)\to\A(N), a\mapsto\A(\phi;a)$.

These two structures satisfy three properties. Associativity means that for every $M$, the map $\A(M)^{\otimes 3}\to\A(M), a\otimes b\otimes c\mapsto a(bc)-(ab)c$ vanishes. Being a homomorphism means that for every $\phi$, the map $\A(M)^{\otimes 2}\mapsto \A(N), a\otimes b \mapsto \A(\phi;ab)-\A(\phi;a)\A(\phi;b)$ vanishes. Functoriality means that for every pair of composable morphisms, $\phi$ and $\psi$, the map $\A(M)\to \A(P), a\mapsto \A(\psi;\A(\phi;a))-\A(\psi\circ\phi;a)$ vanishes.\footnote{Note that I am denoting the linear map as $\A[\phi]$ (with square brackets) and the value of that linear map on $a$ as $\A(\phi;a)$ (with parentheses and a semicolon). This and similar notation will be needed frequently.}

A symmetry\footnote{Fewster \cite{few2013} showed that this is a good definition of symmetry for an LCQFT.} of the functor $\A$ is a natural automorphism $\alpha :\A\naturalto\A$. This is given by, for every object $M$, an automorphism $\alpha[M]:\A(M)\to\A(M)$. This satisfies two properties. Being an automorphism means that for every $M$, the map $\A(M)^{\otimes 2} \to \A(M), a\otimes b \mapsto \alpha(M;ab)-\alpha(M;a)\alpha(M;b)$ vanishes (and that $\alpha[M]$ is invertible). Naturality means that for every $\phi$, the map $\A(M)\to\A(N), a\mapsto \A(\phi;\alpha(M;a))-\alpha(N;\A(\phi;a))$ vanishes.

Each of these structures and conditions depends upon an element of the \emph{nerve}, $B_\bullet\Xc$, of the category $\Xc$. $B_0\Xc = \Obj \Xc$ is the set of objects. $B_1\Xc = \Mor\Xc$ is the set of morphisms. $B_2\Xc$ is the set of composable pairs of morphisms. In general, $B_p\Xc$ is the set of composable $p$-tuples of morphisms. Each element of $B_p\Xc$ begins at an object and ends at an object; for example, $M\in B_0\Xc$ begins and ends at $M$, but $M\stackrel\phi\to N\stackrel\psi\to P$ begins at $M$ and ends at $P$.

Each structure or condition consists of --- for every chain of a given length in $B_\bullet\Xc$ --- a multilinear map from the algebra at the beginning to the algebra at the end. This suggests that the generalization of $C^\bullet(A,A)$ is bigraded. One degree is (again) the multilinearity and the other degree is the chain length (the degree in $B_\bullet\Xc$).

\subsubsection{The Hochschild bicomplex}
\begin{definition}
\begin{align*}
C^{p,q}(\A,\A) &:= \prod_{(M_0\leftarrow\dots\leftarrow M_p)\in B_p\Xc} \Hom_\co[\A(M_p)^{\otimes q},\A(M_0)] \\
&\,= \prod_{(M_0\leftarrow\dots\leftarrow M_p)\in B_p\Xc} C^q[\A(M_p),\A(M_0)]
\end{align*}
\end{definition}
\begin{remark}
I am writing morphisms as arrows from right to left. This is consistent with the usual convention for writing compositions.
\end{remark}

Any chain in $B_p\Xc$ can be composed to a single morphism. Applying $\A$ to this morphism gives a homomorphism from $\A(M_p)$ to $\A(M_0)$, which makes $\A(M_0)$ a bimodule of $\A(M_p)$.
\begin{definition}
The \emph{Hochschild coboundary} $\dH : C^{p,q}(\A,\A)\to C^{p,q+1}(\A,\A)$ is  $(-1)^p$ times the Hochschild coboundary on $C^q[\A(M_p),\A(M_0)]$, i.e.,
\[
\dH = \sum_{i=0}^{q+1}(-1)^{p+i}\dH_i\ .
\]
\end{definition}

The nerve, $B_\bullet\Xc$ is a simplicial set. In particular, there are face maps $\partial_i:B_p\Xc\to B_{p-1}\Xc$, for $0\leq i\leq p$. For $\phi:M\to N$, these are the source and target, $\partial_0(\phi)=M$ and $\partial_1(\phi)=N$. For $p\geq 2$, 
\begin{align*}
\partial_0(\phi_1,\dots,\phi_p) &= (\phi_2,\dots,\phi_p)\\
\partial_i(\phi_1,\dots,\phi_p) &= (\phi_1,\dots,\phi_{i}\circ\phi_{i+1},\dots,\phi_p)\\
\partial_p(\phi_1,\dots,\phi_p) &= (\phi_1,\dots,\phi_{p-1}) \ .
\end{align*}

The face maps correspond to injective maps in the simplicial category. Specifically, $\partial_i$ corresponds to the inclusion of $\{0,\dots,p-1\}$ into $\{0,\dots,p\}$ that skips $i$. Other injective maps can be specified by the numbers that they skip, and the corresponding face maps will be useful. Specifically,
\[
\partial_{i\dots j}(\phi_1,\dots,\phi_p) = (\phi_1,\dots,\phi_{i-1},\phi_i\circ\dots\circ\phi_{j+1},\phi_{j+2},\dots,\phi_p) 
\]
and
\[
\partial_{0\dots i,j\dots p}(\phi_1,\dots,\phi_p) = (\phi_{i+2},\dots,\phi_{j-1}) \ .
\]
There are also degeneracy maps, given by inserting identity morphisms, but these will not be needed.

\begin{definition}
\begin{itemize}
\item[]
\item
For $1\leq i \leq p$, 
\[
\dS_i\Gamma = \Gamma\circ \partial_i \ .
\]  
\item
For $p=0$, 
\begin{gather*}
(\dS_0\Gamma)[\phi] = \A[\phi]\circ \Gamma[M] \ ,\\
(\dS_1\Gamma)[\phi] = \Gamma[N]\circ\A[\phi]^{\otimes q} \ .
\end{gather*} 
\item
For $p\geq 1$, 
\begin{gather*}
(\dS_0\Gamma)[\phi_1,\dots,\phi_{p+1}] = \A[\phi_1]\circ \Gamma[\phi_2,\dots,\phi_{p+1}] \ ,\\
 (\dS_{p+1}\Gamma)[\phi_1,\dots,\phi_{p+1}] = \Gamma[\phi_1,\dots,\phi_p]\circ\A[\phi_{p+1}]^{\otimes q}\ .
\end{gather*}
\end{itemize}
The \emph{simplicial coboundary} $\dS: C^{p,q}(\A,\A)\to C^{p+1,q}(\A,\A)$ is dual to this simplicial structure and is defined by
\[
\dS := \sum_{i=0}^{p+1} (-1)^i\dS_i \ .
\]
\end{definition}

These satisfy $(\dS)^2=(\dH)^2=\dH\dS+\dS\dH=0$, so $C^{\bullet\bullet}(\A,\A)$ with the coboundaries $\dS$ and $\dH$ is a first quadrant bicomplex. 
\begin{definition}
The Hochschild cohomology of a diagram of algebras is 
\[
H^\bullet(\A,\A) := H^\bullet( C^{\bullet}(\A,\A))
\]
where
\[
C^n(\A,\A) := \tot^n C^{\bullet\bullet}(\A,\A) = \bigoplus_{p=0}^n C^{p,n-p}(\A,\A)
\]
with the coboundary $\delta:=\dS+\dH$. Denote the space of closed cycles as $Z^\bullet(\A,\A)$.

Following Gerstenhaber and Schack \cite{gs1988b}, also define the \emph{asimplicial} bicomplex\footnote{They denoted this with an $s$ in \cite{gs1988a}.}
\[
C_a^{p,q}(\A,\A) := 
\begin{cases}
0 & q=0\\
C^{p,q}(\A,\A) & q\geq1
\end{cases}
\]
and from this, $C^\bullet_a$, $Z^\bullet_a$, and $H^\bullet_a$ are defined analogously.
\end{definition} 

\subsubsection{Binary operations}
The naive product of $\Delta\in C^{p,q}(\A,\A)$ and $\Gamma\in C^{p,q'}(\A,\A)$ is $\Delta\cdot \Gamma\in C^{p,q+q'}(\A,\A)$ defined by
\begin{multline*}
(\Delta\cdot \Gamma)(\phi_1,\dots,\phi_p;a_1,\dots,a_q,b_1,\dots,b_{q'}) =\\ \Delta(\phi_1,\dots,\phi_p;a_1,\dots,a_q)\Gamma(\phi_1,\dots,\phi_p;b_1,\dots,b_{q'}) \ .
\end{multline*}

Elements $\Gamma\in C^{p,q}(\A,\A)$ and $\Delta\in C^{p',q'}(\A,\A)$ can be combined by several binary operations.

The cup product, $\Gamma\smile\Delta\in C^{p+p',q+q'}(\A,\A)$ is defined by 
\[
\Gamma\smile\Delta := (-1)^{q p'}\dS_{p+1,\dots,p+p'+1}\Gamma\cdot\dS_{0,\dots,p-1}\Delta \ .
\]

For $1\leq j\leq q$, the partial composition is defined by
\begin{multline*}
(\Gamma\circ_j\Delta)(\sigma;a_1,\dots,a_{q+q-1}) :=\\ \Gamma(\partial_{p+1,\dots,p+p'+1}\sigma;\A(\partial_{p,p+p'+1}\sigma;a_1), \dots,\Delta(\partial_{0,\dots,p-1}\sigma;a_i,\dots,a_{i+q'-1}),\\ \dots,\A(\partial_{p,p+p'+1}\sigma;a_{q+q'-1})) \ .
\end{multline*}
These combine to define $\Gamma\circ\Delta \in C^{p+p',q+q'-1}(\A,\A)$ by 
\[
\Gamma\circ\Delta := \sum_{j=1}^q  (-1)^{(q-1)p'+(q'-1)(q-j)}\Gamma \circ_j\Delta \ .
\]

For $1\leq i\leq p$, there is another kind of ``composition'' defined by $\Gamma\bullet_i\Delta := \dS_{i,\dots,i+p'-2}\Delta\cdot\dS_{0,\dots,i-2,i+p',\dots,p+p'-1}\Gamma$.
Note that $\Gamma$ and $\Delta$ are multiplied in a surprising order. These combine to define $\Gamma\bullet\Delta\in C^{p+p'-1,q+q'}(\A,\A)$ by
\[
\Gamma\bullet\Delta := \sum_{i=1}^p (-1)^{qq'+(p'-1)(p+q-i)}\Gamma\bullet_i\Delta
\]
\begin{remark}
If $q=0$, then $\Delta\circ\Gamma=0$. If $p=0$ or $p'=0$, then $\Delta\bullet\Gamma=0$.
\end{remark}

The analogue of the $\circ$ operation of the ordinary Hochschild complex is
\[
\Gamma\comp\Delta:=\Gamma\circ\Delta+\Gamma\bullet\Delta
\]
and the generalized Gerstenhaber bracket is
\[
[\Gamma,\Delta] := \Gamma\comp\Delta - (-1)^{(p+q-1)(p'+q'-1)}\Delta\comp\Gamma \ .
\]

The cup product and bracket give well defined operations on cohomology, and these make $H^\bullet(\A,\A)$ and  $H^\bullet_a(\A,\A)$ into Gerstenhaber algebras. However, in contrast to the case of a single algebra, this bracket does not make $C^\bullet(\A,\A)$ into a graded Lie algebra.

\subsubsection{Involution}
If $\A:\Xc\to\sAlg$, then there is also an antilinear involution on this bicomplex. 
\begin{definition}
For $\Gamma\in C^{p,q}(\A,\A)$, $\Gamma^\star\in C^{p,q}(\A,\A)$ is given by
\[
\Gamma^\star(\phi_1,\dots,\phi_p;a_1,\dots,a_q) := (-1)^{pq+p(p+1)/2}(\Gamma(\phi_1,\dots,\phi_p;a_q^*,\dots,a_1^*))^*
\]
for any $M_0 \xleftarrow{\phi_1}\dots\xleftarrow{\phi_p} M_p$ 
 and $a_1,\dots,a_q\in \A(M_0)$.
\end{definition}

This operation does not commute with the coboundaries. Instead, 
\[
(\delta \Gamma)^\star = (-1)^{p+q+1}\delta(\Gamma^\star) \ .
\]
This is sufficient to give a well defined involution on cohomology.

This is an involution of the $\smile$ product, up to a homotopy given by the $\bullet$ operation. For $\Gamma\in C^{p,q}(\A,\A)$ and $\Delta\in C^{p',q'}(\A,\A)$,
\begin{multline*}
(\Gamma\smile \Delta)^\star - \Delta^\star\smile \Gamma^\star = \\
(-1)^{(p+q)(p'+q'-1)+1}\left(\delta\Gamma^\star\bullet \Delta^\star - \delta(\Gamma^\star\bullet \Delta^\star) + (-1)^{p+q+1}\Gamma^\star\bullet \delta\Delta^\star\right) \ .
\end{multline*}

\section{The Significance of Hochschild Cohomology}
\label{Significance}
Hochschild cohomology is mainly concerned with infinitesimal things such as derivations and infinitesimal deformations. These concepts are appropriate to purely algebraic AQFT (such as perturbative AQFT) but are not well suited to \cs-algebraic AQFT. However, some of these infinitesimal concepts have finite analogues, which are conceptually clearer and  directly applicable to  \cs-algebraic AQFT.

\subsection{Asimplicial cohomology}
The cohomology $H^\bullet_a(\A,\A)$ is more directly relevant to deformations and symmetries of a diagram of algebras.
\subsubsection{$Z^2_a$}
A diagram of algebras, $\A$, includes two deformable structures. Denote the multiplication in $\A(M)$ as $\m(M;a,b)=ab$; this defines $\m\in  C^{0,2}(\A,\A)$. Also denote $\mu(\phi;a) = \A(\phi;a)$; this defines $\mu\in C^{1,1}(\A,\A)$.

These structures satisfy three conditions, which can be expressed as the vanishing of cochains. Associativity is a condition on $\m$, expressed in $C^{0,3}(\A,\A)$. The compatibility between $\m$ and $\mu$ is that $\mu$ must consist of morphisms between the products given by $\m$; this is expressed in $C^{1,2}(\A,\A)$ and is explicitly
\beq
\label{morphism}
\mu(\phi;\m(M;a,b)) = \m(N;\mu(\phi;a),\mu(\phi;b)) 
\eeq
for $\phi:M\to N$ and $a,b\in\A(M)$. Functoriality is a condition on $\mu$, expressed in $C^{2,1}(\A,\A)$; explicitly,
\beq
\label{functor}
\mu(\psi\circ\phi;a) = \mu(\psi;\mu(\phi;a))
\eeq
for $P\stackrel\psi\leftarrow N\stackrel\phi\leftarrow M$ and $a\in\A(M)$.

Now imagine that $\m+\mu\in C_a^2(\A,\A) = C^{0,2}(\A,\A)\oplus C^{1,1}(\A,\A)$ is part of a 1-parameter family of structures satisfying these conditions. Differentiating these conditions gives 3 conditions on $\dot\m$ and $\dot\mu$. As in the case of a single algebra, differentiating associativity gives the condition $0=\dH\dot\m$.
Differentiating eq.~\eqref{morphism} gives 
\[
\dH_1\dot\mu+\dS_0\dot\m = \dS_1\dot\m + \dH_2\dot\mu+\dH_0\dot\mu\ ,
\] 
thus $0=\dS\dot\m+\dH\dot\mu$. Differentiating the functoriality condition \eqref{functor} gives
\[
\dS_1\dot\mu = \dS_2\dot\mu+\dS_0\dot\mu\ ,
\]
thus $0 = \dS\dot\mu$. Together, these mean that $\dot\m$ and $\dot\mu$ give a cocycle $\dot\m+\dot\mu\in Z^2_a(\A,\A)$.

Differentiating twice gives further conditions. As in the case of a single algebra, the second derivative of associativity is $2\dot\m\circ\dot\m = [\dot\m,\dot\m]=-\dH\ddot\m$. The second derivative of eq.~\eqref{morphism} is
\[
\dH_1\ddot\mu+\dS_0\ddot\m + 2\dot\mu\circ_1\dot\m = \dS_1\ddot\m+\dH_2\ddot\mu+\dH_0\ddot\mu + 2\dot\m\circ_1\dot\mu+2\dot\m\circ_2\dot\mu+2\dot\mu\bullet_1\dot\mu\ ,
\]
thus $2\dot\m\circ\dot\mu+2\dot\mu\bullet\dot\mu = - \dH\ddot\mu-\dS\dot\m$. The second derivative of functoriality \eqref{functor} is
\[
\dS_1\ddot\mu = \dS_2\ddot\mu + \dS_0\ddot\mu + 2\dot\mu\circ_1\dot\mu\ ,
\]
thus $2\dot\mu\circ\dot\mu = -\dS\dot\mu$. Together, these show that 
\[
[\dot\m+\dot\mu,\dot\m+\dot\mu] = -\delta(\ddot\m+\ddot\mu)\ ;
\]
in other words, $[\dot\m+\dot\mu,\dot\m+\dot\mu]$ is exact.

For a diagram of $*$-algebras, $\m$ and $\mu$ satisfy two further conditions. Being a $*$-algebra means that $\m(M;a^*,b^*) = \m(M;b,a)^*$; equivalently, $\m^\star=\m$. Being a $*$-homomorphism means that $\mu(\phi;a^*)=\mu(\phi;a)^*$; equivalently, $\mu^\star=\mu$. The conditions are the same on the derivatives, so $\dot\m+\dot\mu$ is a self-adjoint cocycle.

\subsubsection{$H^2_a$}
The collection of diagrams of algebras, $\Alg^\Xc$ is a category in the usual way, meaning that a homomorphism of diagrams of algebras is a natural transformation. 

Consider a diagram of algebras, $\A$, and let $\B$ be another diagram with the same underlying vector spaces and the other structures denoted by $\m$ and $\mu$. Let $\alpha:\B\to\A$ be a natural isomorphism. This means that for $a,b\in \A(M)$,
\beq
\label{morphism2}
\alpha(M;\m(M;a,b)) = \alpha(M;a)\alpha(M;b) \ ,
\eeq
and for $\phi:M\to N$,
\beq
\label{natural}
\mu(\phi;\alpha(M;a)) = \alpha(N;\A(\phi;a)) \ .
\eeq

Now, imagine that $\alpha$ is part of a 1-parameter family of isomorphisms, starting from the identity. Differentiating eq.~\eqref{morphism2} (and then setting $\alpha=\id$) gives
\[
\dH_1\dot\alpha + \dot\m = \dH_2\dot\alpha+\dH_0\dot\alpha \ ,
\]
thus $\dot\m=\dH\dot\alpha$. Likewise, differentiating the naturality condition \eqref{natural} gives
\[
\dot\mu+\dS_0\dot\alpha = \dS_1\dot\alpha \ ,
\]
thus $\dot\mu = \dS\dot\alpha$. Together, this is $\dot\m+\dot\mu = \delta\dot\alpha$. In other words,  a trivial deformation is given to first order by an exact cocycle.

A class in $H^2_a(\A,\A)$ is called a \emph{Maurer-Cartan} element if the bracket with itself is $0$. 
Deformations of $\A$ are classified to first order modulo trivial deformations by the Maurer-Cartan elements.

\subsubsection{$H^1_a$}
The symmetries of $\A$ are its natural automorphisms, so now consider $\alpha$ a natural automorphism of $\A$. This means that for $a,b\in \A(M)$,
\beq
\label{morphism3}
\alpha(M;ab) = \alpha(M;a)\alpha(M;b) \ ,
\eeq
and for $\phi:M\to N$,
\beq
\label{natural2}
\A(\phi;\alpha(M;a)) = \alpha(N;\A(\phi;a)) \ .
\eeq

Suppose that $\alpha$ is part of a 1-parameter family of natural automorphisms, starting from the identity. Differentiating eq.~\eqref{morphism3} (and setting $\alpha=\id$) gives
\[
\dH_1\dot\alpha = \dH_2\dot\alpha+\dH_0\dot\alpha \ ,
\]
thus $0=\dH\dot\alpha$. Differentiating eq.~\eqref{natural2} gives
\[
\dS_0\dot\alpha=\dS_1\dot\alpha\ ,
\]
thus $0=\dS\dot\alpha$. Together, this gives $0=\delta\dot\alpha$, but since $C^0_a(\A,\A)=0$, this means that $\dot\alpha\in H^1_a(\A,\A)$.

The set of natural automorphisms of $\A$ is a group. For two natural automorphisms, $\alpha$ and $\beta$, their product is simply $\alpha\circ\beta=\alpha\circ_1\beta$, given by $(\alpha\circ\beta)(M;a) = \alpha(M;\beta(M;a))$.

The Lie algebra of the group of natural automorphisms is the space $H^1_a(\A,\A)$, with the $\circ$-commutator. This operation is the Gerstenhaber bracket in this degree.

For a diagram of $*$-algebras, we should also require that $\alpha(M;a^*)=\alpha(M;a)^*$. This just means that $\alpha^\star=\alpha$. The condition on the derivative is the same. The Lie algebra of infinitesimal natural $*$-automorphisms is thus the $\star$-invariant part of $H^1_a(\A,\A)$.

\subsection{Full Hochschild cohomology}
\subsubsection{$Z^2$}
A cocycle in the full Hochschild complex contains more structure. This complex should describe the deformations of a diagram of algebras to something more general.

The additional structure is an element $\xi\in C^{2,0}(\A,\A)$. For any $P\stackrel\psi\leftarrow N\stackrel\phi\leftarrow M$, this gives some $\xi(\psi,\phi)\in\A(P)$. For a diagram of $*$-algebras, if $\xi^\star=\xi$, then $\xi(\psi,\phi)$ is \emph{anti-selfadjoint}; this suggests that this is the first order part of a unitary. In general, we should have an invertible element $u(\psi,\phi)\in 1 + \A(P) \subset \tilde\A(P)$ (the unitalization of $\A(P)$).

Let $\m+\mu+u$ be an example of this (yet undetermined) generalization of a diagram of algebras. Suppose that if this is part of an 1-parameter family of such structures (starting from $\A$ with $u=1$) then the first derivative satisfies $0=\delta(\dot\m+\dot\mu+\dot u)$ and
\[
[\dot\m+\dot\mu+\dot u,\dot\m+\dot\mu+\dot u]
\]
is exact. These conditions have components in degrees $(0,3)$, $(1,2)$, $(2,1)$, and $(3,0)$.

The conditions in degrees $(0,3)$ and $(1,2)$ are unchanged, thus we should still require that $\m[M]$ is an associative product, and $\mu[\phi]$ is a homomorphism. 

The condition in degree $(2,1)$ is modified by terms involving $u$. This means that the functoriality condition \eqref{functor} on $\mu$ must be modified by $u$. The second derivative of this condition should be, $0=\dH\ddot\mu+\dS\ddot u + 2\dot\m\circ\dot u + 2\dot\mu\circ\dot\mu+2\dot\mu\bullet\dot u + 2 \dot u\bullet\dot\mu$, which is
\[
\dS_0\ddot\mu+\dS_2\ddot\mu + \dH_0\ddot u + 2\dot\m\circ_2\dot u + 2\dot\mu\circ_1\dot\mu + 2\dot u\bullet_1\dot\mu + 2\dot u\bullet_2\dot\mu = \dS_1\ddot\mu+\dH_1\ddot u + 2\dot\m\circ_1\dot u+ 2\dot\mu\bullet_1\dot u \ .
\]

The condition in degree $(3,0)$ is new and only involves $u$. This suggests that $u$ satisfies some nonlinear cocycle condition indexed by $B_3\Xc$. 
The second derivative of this condition should be $0=\dS\ddot u + 2\dot u\bullet\dot u$, which is
\[
\dS_0\ddot u + \dS_2\ddot u + 2\dot u\bullet_2\dot u = \dS_1\ddot u + \dS_3\ddot u+ 2\dot u\bullet_1\dot u\ .
\]

This leads to the following definition.
\begin{definition}
\label{Skew}
A \emph{skew diagram of algebras} $(\A,u)$ over a category $\Xc$ consists of
\begin{itemize}
\item
for every $M\in\Xc$, an associative algebra $\A(M)$,
\item
for every $\phi:M\to N$ in $\Xc$, a homomorphism $\A[\phi]:\A(M)\to\A(N)$, and
\item
for every $P\stackrel\psi\leftarrow N\stackrel\phi\leftarrow M$, an invertible element $u(\psi,\phi)\in 1+\A(P)$,
\end{itemize}
such that:
\begin{itemize}
\item
For any $P\stackrel\psi\leftarrow N\stackrel\phi\leftarrow M$ and $a\in\A(M)$,
\[
\A(\psi;\A(\phi;a))u(\psi,\phi) = u(\psi,\phi)\A(\psi\circ\phi;a)\ ;
\]
\item
for any $Q\stackrel\chi\leftarrow  P\stackrel\psi\leftarrow N\stackrel\phi\leftarrow M$,
\[
\A(\chi;u(\psi,\phi)) u(\chi,\psi\circ\phi) = u(\chi,\psi)u(\chi\circ\psi,\phi)\ .
\]
\end{itemize}
\end{definition}

\begin{remark}
The penultimate condition can be stated as commutativity of the diagram,
\[
\begin{tikzcd}
\A(M) \arrow{r}{\A[\psi\circ\phi]} \arrow{d}{\A[\phi]} & \A(P) \arrow{d}{\Ad[u(\psi,\phi)]} \\
\A(N) \arrow{r}{\A[\psi]} & \A(P)
\end{tikzcd}
\]
where $\Ad[u]:a\mapsto u a u^{-1}$.
\end{remark}

This definition extends easily to some other categories of algebras.
\begin{definition}
For a skew diagram of $*$-algebras, $u(\psi,\phi)$ is required to be unitary. For a skew diagram of unital algebras, $u(\psi,\phi)\in \A(M)$.
\end{definition}
\begin{remark}
Note that a diagram of algebras is precisely a skew diagram of the form $(\A,1)$.
\end{remark}

By construction, if a diagram of algebras, $\A$, is deformed to a 1-parameter family of skew diagrams of algebras, then at first order the deformation determines an element of $Z^2(\A,\A)$ whose bracket with itself is exact.

\subsubsection{$H^2$}
By analogy with $H^2_a(\A,\A)$, the Maurer-Cartan elements of $H^2(\A,\A)$ should classify the deformations of $\A$ into skew diagrams up to first order, modulo trivial deformations. A trivial deformation should mean one in which all of the skew diagrams are isomorphic to $\A$. This guides us to define morphisms of skew diagrams.

If $\A$ is a diagram of algebras, then a trivial deformation of $\A$ to skew diagrams can be constructed using a 1-parameter family of isomorphisms. This is given to first order by an element of $C^{0,1}(\A,\A)\oplus C^{1,0}(\A,\A)$. The first part is not new, and corresponds to a family of homomorphisms (indexed by the objects of $\Xc$). The second part is a family of algebra elements indexed by morphisms in $\Xc$. In the case of $*$-algebras, a $\star$-invariant element of $C^{1,0}(\A,\A)$ is a family of anti-self-adjoint algebra elements. This corresponds to a family of unitary (or generally, invertible) elements of the unitalized algebras.

We don't yet know what a morphism of skew diagrams is, but suppose that $(\alpha,v):(\B,u)\to(\A,1)$ is an isomorphism of skew diagrams, where $\B$ has the same underlying vector spaces.
 The structures $\m$, $\mu$, and $u$ of this $(\B,u)$ should be determined by some formulae from $\alpha$, $v$, and the structures of $\A$. As before, the product is defined by 
\[
\alpha(M;\m(M;a,b)) = \alpha(M;a)\alpha(M;b) \ ,
\]
for $a,b\in\A(M)$

Suppose that $(\alpha,v)$ is part of a 1-parameter family, starting from $\alpha=\id$ and $v=1$. To first order, $\mu$ is given by $\dot\mu = \dS\dot\alpha+\dH\dot v = \dS_0\dot\alpha - \dS_1\dot\alpha-\dH_0\dot v+\dH_1\dot v$; more explicitly,
\[
\dot\mu(\phi;a)+ \dot\alpha(N;\A(\phi;a)) + \A(\phi;a) \dot v(\phi) = \A(\phi;\dot\alpha(M;a))+\dot v(\phi)\A(\phi;a) \ ,
\]
for $\phi:M\to N$ and $a\in\A(M)$.
This is the derivative of 
\[
\alpha(N;\mu(\phi;a)) v(\phi) = v(\phi) \A(\phi;\alpha(M;a)) \ ,
\]
which defines $\mu$.

To first order, $u$ is given by $\dot u = \dS\dot v = \dS_0\dot v - \dS_1\dot v + \dS_2\dot v$. More explicitly,
\[
\dot u(\psi,\phi) + \dot v(\psi\circ\phi) = \A(\psi;\dot v(\phi)) + \dot v(\psi) \ ,
\]
for $M\stackrel\phi\to N \stackrel\psi\to P$. This is the first derivative of 
\[
\alpha(P;u(\psi,\phi)) v(\psi\circ\phi) = v(\psi) \A(\psi;v(\phi)) \ ,
\]
which defines $u$.
\begin{remark}
The orderings in the products in this formula are not obvious. This is linked with the reversed order of multiplication in the definition of $\bullet$. 

In fact, there are 2 possible conventions. The definition of a skew diagram of algebras can be changed so that the multiplication here and in the definition of $\bullet$ is in the obvious order. The repercussion of this choice is that the definition of the cup product would be more awkward. I am using the convention consistent with Gerstenhaber and Schack \cite{gs1988a,gs1988b}.
\end{remark}

Further extrapolation leads to the following definition.
\begin{definition}
\label{Skew morphism}
A \emph{morphism of skew diagrams of algebras} $(\alpha,v) : (\A,u)\to (\B,u')$ is given by 
\begin{itemize}
\item
For every object $M\in\Obj\Xc$, a homomorphism $\alpha[M]:\A(M)\to\B(M)$, and
\item
for every morphism $\phi:M\to N$, an element $v(\phi)\in 1 + \B(N)$,
\end{itemize}
such that
\begin{itemize}
\item
for $\phi:M\to N$ and $a\in\A(M)$,
\beq
\label{unnatural}
\alpha(N;\A(\phi;a))\, v(\phi) = v(\phi)\, \B(\phi;\alpha(M;a)) 
\eeq
\item
and for $M\stackrel\phi\to N \stackrel\psi\to P$,
\beq
\label{morphism condition}
\alpha(P;u(\psi,\phi))\, v(\psi\circ\phi) = v(\psi)\, \B(\psi;v(\phi))\, u'(\psi,\phi)\ .
\eeq
\end{itemize}
For a morphism of skew diagrams of unital algebras, $v(\phi)\in \B(N)$, and for a morphism of skew diagrams of $*$-algebras it should be unitary.
\end{definition}
\begin{remark}
Equation \eqref{unnatural} generalizes the definition of a natural transformation. It can be expressed as commutativity of the diagram
\[
\begin{tikzcd}
\A(M) \arrow{rr}{\A[\phi]} \arrow{dd}{\alpha[M]} 
&& \A(N) \arrow{d}{\alpha[N]}   \\
&& \B(N) \arrow{d}{\cdot v(\phi)}\\
\B(M) \arrow{r}{\B[\phi]} & \B(N) 
\arrow{r}{v(\phi)\cdot} & \B(N)
\end{tikzcd}
\]
which is like the diagram defining naturality, but with the lower right corner modified.
\end{remark}

\subsubsection{$Z^1$}
A diagram of algebras $\A$ is in particular a skew diagram, if we set $u=1$. We can thus consider the group, $\SAut(\A)$, of skew automorphisms of $\A$, i.e., automorphisms of $(\A,1)$ in the category of skew diagrams of algebras (which I have not yet fully defined).

The defining properties of a skew automorphism $(\alpha,v)$ of $\A$ simplify to: $\alpha[M]$ is a homomorphism,
\[
\alpha(N;\A(\phi;a))\, v(\phi) = v(\phi)\, \A(\phi;\alpha(M;a))\ ,
\]
and
\[
v(\psi\circ\phi) = v(\psi)\, \A(\psi;v(\phi))\ .
\]
For a 1-parameter family of such automorphisms, starting from $(\id,1)$, the first derivatives of these conditions give precisely that $\dot\alpha+\dot v$ is closed. So, $Z^1(\A,\A)$ is the space of infinitesimal skew automorphisms.

The Gerstenhaber bracket on $C^\bullet(\A,\A)$ does not satisfy the graded Jacobi identity. Nevertheless, $Z^1(\A,\A)$ actually is a Lie algebra. So, let's compute the Gerstenhaber bracket of $\xi+\upsilon$ and $\xi'+\upsilon'\in C^{0,1}(\A,\A)\oplus C^{1,0}(\A,\A)$. Because of the low degrees, $\xi\bullet \upsilon' = 0$. The only degree $(0,1)$ part of $(\xi+\upsilon)\comp(\xi'+\upsilon')$ is $\xi\circ\xi' = \xi\circ_1\xi'$, given by 
\[
(\xi\circ\xi')(M;a) = \xi(M;\xi'(M;a))\ .
\]
The degree $(1,0)$ part is $\upsilon\bullet\upsilon'+\xi\circ\upsilon' = \upsilon\bullet_1\upsilon'+\xi\circ_1\upsilon'$, given by 
\[
(\upsilon\bullet\upsilon'+\xi\circ\upsilon')(\phi) = \upsilon'(\phi)\,\upsilon(\phi) + \xi(N;\upsilon(\phi))\ .
\]
(Note the order of multiplication.) The Gerstenhaber bracket is given by antisymmetrizing this. 

The Gerstenhaber bracket actually satisfies the Jacobi identity on the larger subspace of $\xi+\upsilon\in C^1(\A,\A)$ such that $\dH\xi=0$ (i.e., $\xi[M]$ is a derivation).
This Lie algebra has subalgebras indexed by the objects and morphisms of $\Xc$. For every object $M\in\Obj\Xc$, there is the Lie algebra $\der(\A(M))$ of derivations. For every morphism $\phi:M\to N$, there is a copy of $\A(N)$ with bracket equal to \emph{minus} the commutator. Obviously, $\der(\A(M))$ acts on $\A(M)$ by derivations. This Lie algebra is the semidirect product of all of the subalgebras. 

The corresponding group has subgroups indexed by the objects and morphisms in $\Xc$. For every object, $M$, there is the group of automorphisms, $\Aut(\A(M))$. For every morphism with codomain $M$, there is the group of invertible elements in $1+\A(M)^\op$. The group of skew automorphisms of $\A$ is a subgroup of the semidirect product of these groups.

This leads to the definition of composition.
\begin{definition}
\label{Composition}
If $(\alpha,v):(\A,u)\to (\B,u')$ and $(\beta,v'):(\B,u')\to(\Cf,u'')$ are morphisms of skew diagrams of algebras, then their composition $(\gamma,v'') = (\beta,v')(\alpha,v) : (\A,u)\to (\Cf,u'')$ is given by
\beq
\label{composition1}
\gamma(M;a)=\beta(M;\alpha(M;a))
\eeq
and
\beq
\label{composition2}
v''(\phi) = \beta(N;v(\phi))\,v'(\phi) \ .
\eeq
\end{definition}
\begin{remark}
There are two natural ways of identifying the semidirect product of groups with the Cartesian product as a set. The other choice is not compatible with the fact that $v(\phi)\in 1+\A(N)$.
\end{remark}

\begin{thm}
With these definitions of the objects, morphisms, and compositions, skew diagrams of algebras over $\Xc$ form a category.
\end{thm}
The proof is a straightforward calculation. It also follows from Theorem~\ref{Higher thm}, below

\subsubsection{$H^1$}
As I have noted, $Z^1(\A,\A)$ is a Lie algebra under the Gerstenhaber bracket. It acts on $C^0(\A,\A)$ and on itself by the Gerstenhaber bracket, and the coboundary
\[
\delta : C^0(\A,\A) \to Z^1(\A,\A)
\]
is equivariant, so its image is an ideal, and thus $H^1(\A,\A)$ is a Lie algebra. The obvious commutator bracket makes  $C^0(\A,\A)$ a Lie algebra and $-\delta$ a homomorphism.
This suggests that the analogue of $H^1(\A,\A)$ should be a quotient of the group  $\SAut(\A)$. 

Let $\xi+\upsilon\in Z^1(\A,\A)$ and $\zeta\in C^0(\A,\A)$. The Gerstenhaber bracket in these degrees reduces to 
\[
[\xi+\upsilon,\zeta] = \xi\circ\zeta = \xi\circ_1\zeta ,
\]
so the action of $(\alpha,v)\in \SAut(\A)$ on $\zeta$ is also just $\alpha\circ_1\zeta$, i.e., 
\[
(\alpha\circ_1\zeta)(M) = \alpha(M;\zeta(M)) \ .
\]

The finite analogue of $C^0(\A,\A)$ is the set of functions $w\in C^0(\A,\tilde\A)$ such that $w(M)\in 1 + \A(M)$ is invertible; this is a group under the naive product. The finite analogue of $-\delta$ should be an equivariant homomorphism from this group to $\SAut(\A)$. Let $(\alpha,v)$ be the image of $w$ by this homomorphism. If $w$ is part of a 1-parameter family starting from $w=1$, then the first derivative should be given by minus the coboundary, i.e., $\dot\alpha = -\dH\dot w = -\dH_0\dot w + \dH_1\dot w$ and $\dot v = -\dS\dot w = -\dS_0\dot w + \dS_1\dot w$. Explicitly,
\[
\dot\alpha(M;a) = \dot w(M)\, a - a\,\dot w(M) 
\]
and
\[
\dot v(\phi) = \dot w(N) - \A(\phi;\dot w(M))  \ .
\]
These requirements lead uniquely to the formulae,
\[
\alpha(M;a) = w(M)a\, w(M)^{-1}
\]
and
\[
v(\phi) = w(N)\A(\phi;w(M)^{-1}) \ .
\]

\begin{definition}
\label{Outer}
The group of \emph{outer skew automorphisms}, $\SOut(\A)$, is the quotient of $\SAut(\A)$ by this image.
\end{definition}
This is the finite analogue of $H^1(\A,\A)$.  I discuss this further in Section~\ref{Higher}.\label{SOut}

\subsubsection{$H^0$}
In this case the infinitesimal and finite concepts are the same.

$H^0(\A,\A)$ is the joint kernel of $\dS$ and $\dH$ in 
\[
C^{0,0}(\A,\A) = \prod_{M\in \Obj\Xc} \A(M)
\ .
\]
Let $\zeta\in C^{0,0}(\A,\A)$. 

The Hochschild coboundary $\dH\zeta \in C^{0,1}(\A,\A)$ is defined by, for $M\in \Obj\Xc$ and $a\in\A(M)$,
\[
(\dH\zeta)(M;a) = a\,\zeta(M) - \zeta(M) a = [a,\zeta(M)] \ .
\]
The condition $0 = \dH\zeta$ means precisely that $\zeta(M)\in\A(M)$ is central for all $M\in \Obj\Xc$.

The simplicial coboundary $\dS\zeta\in C^{1,0}(\A,\A)$ is defined by, for $\phi:M\to N$ in $\Xc$,
\[
(\dS\zeta)(\phi) = \A(\phi;\zeta(M)) - \zeta(N) \ .
\]
The condition that $0 = \dS\zeta$ is a sort of invariance. 
\begin{example}
If $\Xc=G$ is a group (viewed as a category with one object, $*\in\Obj(\Xc)$) then $\A:G\to\Alg$ is equivalent to a single algebra $A=\A(*)$ with an action of $G$ by automorphisms of $A$. The condition that $0=\dS\zeta$ means that $\zeta$ is $G$-invariant, and $0=\dH\zeta$ means that $\zeta$ is central, so 
\[
H^0(\A,\A) = \Zcenter(A)^G\ .
\]
\end{example}

For a diagram of $*$-algebras, $\zeta^\star=\zeta$ precisely if $\zeta(M)$ is self-adjoint.

For $\A:\Loc\to\sAlg$ a LCQFT, $\zeta=\zeta^\star\in H^0(\A,\A)$ is an observable whose values are unaffected by the action of other observables. Moreover, it can be measured in any arbitrarily small region of spacetime.

\subsection{Higher categorical interpretation}
For the definitions of 2-categories and their functors, transformations, and modifications, see \cite{lei1998,lei2004}.
\begin{definition}
\label{Higher}
Let $\Algt$ be the strict 2-category whose underlying 1-category is $\Alg$ and 
\begin{itemize}
\item
for any homomorphisms $\alpha,\beta:A\to B$, a 2-morphism $\alpha\stackrel{u}{\Longrightarrow}\beta$ is $u\in 1+B$ such that for any $a\in A$,
\[
u \alpha(a) = \beta(a)u \ ;
\]
\item
the horizontal composition of $\alpha\stackrel{u}{\Longrightarrow}\beta\stackrel{v}{\Longrightarrow}\gamma$ is $\alpha\stackrel{vu}{\Longrightarrow}\gamma$ (i.e., $v\circ u=vu$);
\item
the vertical composition of $\alpha\stackrel{u}{\Longrightarrow}\beta$ and $\gamma\stackrel{v}{\Longrightarrow}\delta$ is $v\bullet u :=v\,\gamma(u) = \delta(u)\, v : (\gamma\circ\alpha)\Longrightarrow(\delta\circ\beta)$.
\end{itemize}
\end{definition}

As a category, $\Xc$ is in particular a strict 2-category, with only identity 2-morphisms. Both $\Xc$ and $\Algt$ are in particular weak 2-categories (bicategories). 

\begin{thm}
\label{Higher thm}
The category of skew diagrams of algebras over $\Xc$ is the category whose objects are pseudofunctors from $\Xc$ to $\Algt$ (written as oplax functors) and whose morphisms are lax transformations. 
\end{thm}
\begin{proof}
An oplax functor $(\A,u): \Xc\to\Algt$ maps objects to objects, morphisms to morphisms, and composable pairs to 2-morphisms. In a pseudofunctor, these 2-morphisms are required to be invertible. 
 For each $M\in\Obj\Xc$, it gives an algebra $\A(M)$. For $\phi:M\to N$ it gives a homomorphism $\A[\phi]:\A(M)\to\A(N)$. For $(\psi,\phi)\in B_2\Xc$, it gives  $u(\psi,\phi) : \A[\psi\circ\phi]\Rightarrow \A[\psi]\circ\A[\phi]$. 
 
 By the definition of $\Algt$, this means that $u(\psi,\phi)\in 1+\A(P)$ and 
\[
u(\psi,\phi)\A(\psi\circ\phi;a) = \A(\psi;\A(\phi;a))u(\psi,\phi) \ .
\]
These are the components of a natural transformation, but because $\Xc$ has only identity 2-morphisms, the naturality condition is trivial.

For any $(\chi,\psi,\phi)\in B_3\Xc$, there is a commutative diagram of 2-morphisms,
\[
\begin{tikzcd}
\A[\chi\circ\psi\circ\phi] \arrow[Rightarrow]{rr}{u(\chi,\psi\circ\phi)} \arrow[Rightarrow]{d}{u(\chi\circ\psi,\phi)} &{}& \A[\chi]\circ\A[\psi\circ\phi] \arrow[Rightarrow]{d}{\id_{\A[\chi]}\circ u(\psi,\phi)} \\
\A[\chi\circ\psi]\circ\A[\phi] \arrow[Rightarrow]{rr}{u(\chi,\psi)\circ\id_{\A[\phi]}} &{}& \A[\chi]\circ\A[\psi]\circ\A[\phi]
\end{tikzcd}
\]
that is,
\[
[\id_{\A(\chi)}\circ u(\psi,\phi)]\bullet u(\chi,\psi\circ\phi) = [u(\chi,\psi)\circ\id_{\A(\phi)}]\bullet u(\chi\circ\psi,\phi) \ .
\]
(Here, $\id$ denotes the identity 2-morphism over a 1-morphism in $\Algt$.) The definition of horizontal composition in $\Algt$ gives that
\[
\id_{\A(\chi)}\circ u(\psi,\phi) = \A(\chi;u(\psi,\phi)) : \A(\chi)\circ\A(\psi,\phi) \Longrightarrow \A(\chi)\circ\A(\psi)\circ\A(\phi)
\]
and
\[
u(\chi,\psi)\circ\id_{\A(\phi)} = u(\chi,\psi) : \A(\chi\circ\psi)\circ\A(\phi) \Longrightarrow \A(\chi)\circ\A(\psi)\circ\A(\phi) \ .
\]
The definition of vertical composition in $\Algt$ simplifies this to
\[
\A(\chi;u(\psi,\phi)) u(\chi,\psi\circ\phi) = u(\chi,\psi)u(\chi\circ\psi,\phi)\ .
\]
This shows that a pseudofunctor $(\A,u): \Xc\to\Algt$ is a skew diagram, and that a skew diagram is a pseudofunctor.

Let $(\A,u)$ and $(\B,u')$ be diagrams of algebras over $\Xc$. As we have just seen, these are pseudofunctors $\Xc\to\Algt$, written as oplax functors.
 A lax  transformation  $(\alpha,v):(\A,u)\naturalto(\B,u')$ consists of, for every object $M\in\Obj\Xc$, a homomorphism $\alpha[M]:\A(M)\to\B(M)$, and for every 1-morphism $\phi:M\to N$, a 2-morphism $v(\phi) :\B[\phi]\circ\alpha[M]\Rightarrow\alpha[N]\circ\A[\phi]$. This means that $v(\phi)\in1+\B(N)$ and for any $a\in\A(M)$,
\[
 \alpha(N;\A(\phi;a))\, v(\phi) = v(\phi)\, \B(\phi;\alpha(M;a)) \ .
\]
This is precisely eq.~\eqref{unnatural}. 

For $P\stackrel\psi\leftarrow N\stackrel\phi\leftarrow M$, there is a commutative diagram of 2-morphisms,
\[
\begin{tikzcd}
\B[\psi\circ\phi]\circ\alpha[M] \arrow[Rightarrow]{rr}{v(\psi\circ\phi)} \arrow[Rightarrow]{d}{u'(\psi,\phi)\circ\id_{\alpha[M]}} && \alpha[P]\circ\A[\psi\circ\phi] \arrow[Rightarrow]{d}{\id_{\alpha[P]}\circ u(\psi,\phi)} \\
\B[\psi]\circ\B[\phi]\circ\alpha[M] \arrow[Rightarrow]{dr}{\id_{\B[\psi]}\circ v(\phi)} && \alpha[P]\circ\A[\psi]\circ\A[\phi] \\
& \B[\psi]\circ\alpha[N]\circ\A[\phi] \arrow[Rightarrow]{ur}{v(\psi)\circ\id_{\A[\phi]}}
\end{tikzcd}
\]
that is,
\[
[u'(\psi,\phi)\circ\id_{\alpha[M]}]\bullet v(\psi\circ\phi) = [v(\psi)\circ\id_{\A[\phi]}] \bullet [\id_{\B[\psi]}\circ v(\phi)] \bullet [u'(\psi,\phi)\circ\id_{\alpha[M]}] \ .
\]
By the definitions of horizontal and vertical composition, this is precisely  eq.~\eqref{morphism condition}.

Finally, consider two lax transformations $(\alpha,v):(\A,u)\naturalto (\B,u')$ and $(\beta,v'):(\B,u')\naturalto(\Cf,u'')$, and denote their compositon as $(\gamma,v'')  : (\A,u)\naturalto (\Cf,u'')$. 
 At $M\in\Obj\Xc$, $\gamma[M] = \beta[M]\circ\alpha[M]$; this is eq.~\eqref{composition1}. For $\phi:M\to N$, $v''(\phi):\Cf[\phi]\circ\gamma[M]\Rightarrow \gamma[N]\circ\A[\phi]$ is diagrammatically
\[
\begin{tikzcd}
& \B(N) \arrow[swap]{ld}{\beta[N]} & \A(N) \arrow[swap]{l}{\alpha[N]}  & \\
\Cf(N) &{}&{}& \A(M) \arrow[swap]{ul}{\A[\phi]} \arrow{ld}{\alpha[M]}\\
& \Cf(N) \arrow{lu}{\Cf[\phi]} \arrow[Rightarrow,yshift=3.5ex]{u}{v'(\phi)} & \B(M) \arrow{l}{\beta[M]} \arrow[Rightarrow,yshift=3.5ex]{u}{v(\phi)} \arrow{uul}{\B[\phi]}
\end{tikzcd}
\]
and explicitly
\begin{align*}
v''(\phi) &= [\id_{\beta[N]}\circ v(\phi)]\bullet[v'(\phi)\circ\id_{\alpha[M]}] \\
&= \beta(N;v(\phi))\, v'(\phi)\ ,
\end{align*}
which is eq.~\eqref{composition2}.
\end{proof}

\begin{remark}
There are a few minor variations possible on these definitions, depending upon what is required to be invertible and lax versus oplax versions.
\end{remark}

From this perspective, an additional structure becomes apparent: The category of skew diagrams is actually a strict 2-category. 
The 2-morphisms are modifications between the lax transformations. 

Again, let $(\A,u),(\B,u'):\Xc\to\Algt$ and $(\alpha,v),(\beta,v'):(\A,u)\naturalto(\B,u')$. A \emph{modification} $w:(\alpha,v)\modification(\beta,v')$ consists of, for every $M\in\Obj\Xc$, a 2-morphism $w(M):\alpha[M]\Rightarrow\beta[M]$, such that for any $\phi:M\to N$, there is a commutative diagram of 2-morphisms,
\[
\begin{tikzcd}
\B[\phi]\circ\alpha[M] \arrow[Rightarrow]{rr}{\id_{\B[\phi]}\circ w[M]} \arrow[Rightarrow]{d}{v(\phi)} &{}& \B[\phi]\circ\beta[M] \arrow[Rightarrow]{d}{v'(\phi)} \\
\alpha[N]\circ\A[\phi] \arrow[Rightarrow]{rr}{w(N)\circ\id_{\A[\phi]}} &{}& \beta[N]\circ\A[\phi]
\end{tikzcd}
\]
Applying the definition of $\Algt$ makes this explicit:
\begin{definition}
Given two 1-morphisms of skew diagrams, $(\alpha,v),(\beta,v'):(\A,u)\to(\B,u')$, a \emph{2-morphism} $w:(\alpha,v)\Rightarrow(\beta,v')$ consists of  $w(M)\in1+\B(M)$ (for every object $M$) satisfying 
\beq
\label{modification1}
w(M)\alpha(M;a) = \beta(M;a)w(M)
\eeq
for all $a\in\A(M)$, and
\beq
\label{modification2}
v'(\phi)\B(\phi;w(M)) = w(N)v(\phi) \ ,
\eeq 
for every $\phi:M\to N$.

For $w:(\alpha,v)\Rightarrow (\beta,v')$ and $w':(\beta,v')\Rightarrow (\gamma,v'')$, the vertical composition $w'\bullet w : (\alpha,v)\Rightarrow (\gamma,v'')$ is 
\[
(w'\bullet w)(M) = w'(M)\bullet w(M) = w'(M)w(M) \ .
\]

For three skew diagrams, $(\A,u)$, $(\B,u')$, and $(\Cf,u'')$, four 1-morphisms $(\alpha,v)$, $(\beta,v'):(\A,u)\to(\B,u')$ and $(\delta,v'')$, $(\gamma,v'''):(\B,u')\to(\Cf,u'')$, and two 2-morphisms, $w:(\alpha,v)\Rightarrow(\beta,v')$ and $w':(\delta,v'')\Rightarrow(\gamma,v''')$. The horizontal composition $w'\circ w$ is
\[
(w'\circ w)(M) = w'(M)\circ w(M) = w'(M) \gamma(M;w(M)) \ .
\]
\end{definition}

In particular, for a given diagram $\A$, there is a 2-group of automorphisms, which can be described by a crossed module involving the structures discussed in Section~\ref{SOut}. The crossed module consists of:
\begin{itemize}
\item
the finite analogue of $C^0(\A,\A)$, i.e., the group of invertible elements of $C^0(\A,\tilde\A)$ of the form $w(M)\in1+\A(M)$,
\item
the finite analogue of $Z^1(\A,\A)$, i.e., $\SAut(A)$,
\item
the finite analogue of $-\delta : C^0(\A,\A) \to Z^1(\A,\A)$, and
\item
the finite analogue of the action of $Z^1(\A,\A)$ on $C^0(\A,\A)$ by the Gerstenhaber bracket.
\end{itemize}
It is easy to check that in Section \ref{SOut}, $w:(\id,1)\Rightarrow (\alpha,v)$. 

This gives interpretations of the finite analogues of $H^1(\A,\A)$ and $H^0(\A,\A)$. $\SOut(\A)$ is the $0$'th homotopy group of the automorphism 2-group of $\A$. The group of invertible elements of $H^0(\A,\A)$ is the first homotopy group of the automorphism 2-group of $\A$.

\subsection{Generalized AQFT}
Definition \ref{Skew} of a skew diagram of algebras suggests a generalization of algebraic quantum field theory. With $\Xc=\Loc$ or the category of causally complete regions of a fixed spacetime, we can simply define a generalized AQFT as a skew diagram of $*$-algebras over $\Xc$. Einstein causality, the time-slice axiom, and isotony can be required just as before.
This sets AQFT models within a larger class of structures. 

Given a skew diagram $(\A,u)$ of algebras over $\Xc$, there exists another category $\mathcal Y$, a functor $\pi:\mathcal Y\to\Xc$, a diagram $\B:\mathcal Y\to\Alg$, and a (non-functorial) section $\sigma$ of $\pi$ such that $\A=\B\circ\sigma$. 

If, in some attempt to construct an AQFT, it is only possible to construct a generalized AQFT in this sense, then this is an indication that some additional structure is required beyond the globally hyperbolic spacetimes in $\Loc$.

For example, a model with a spin-$\frac12$ field cannot be formulated as a LCQFT over $\Loc$ (see remarks following Cor.~13 in \cite{few2017}) but can be formulated over the category of globally hyperbolic spin-manifolds. It appears likely that such a model \emph{can} be formulated as a generalized AQFT with a skew diagram over $\Loc$.

This situation is extremely similar to that considered in \cite{b-s}, and the relationship deserves further investigation. Can a quantum field theory on a category $\mathcal Y$ fibered in groupoids over $\Loc$ be described by a skew diagram over $\Loc$? Can a generalized AQFT  be described by a QFT over a fibered category? The answers are almost certainly yes in some cases.

\section{Interaction}
\label{Interaction}
As I have mentioned, an AQFT (or any diagram of algebras) has two deformable structures: the associative products and the maps between algebras. These are elements of the Hochschild bicomplex in degrees $(0,2)$ and $(1,1)$, respectively. This means that there are 2 qualitatively different ways of deforming an AQFT. For example:
\begin{itemize}
\item
The transition from a classical to a quantum field theory deforms the associative algebra structures.
\item
The transition from a free to an interacting field theory deforms the maps between algebras.
\end{itemize}

It may not always be possible to disentangle these 2 aspects of deformation, but it is known that the von~Neumann algebra associated to any connected, precompact region of spacetime is isomorphic to the (unique) hyperfinite type III${}_1$ factor \cite{bdf1987,fre1986,fre1987}. (This does not apply to classical field theories.) This strongly indicates that deformation of maps is far more important to the deformation of a quantum field theory.

If an  AQFT, $\A$, is deformed smoothly by changing the interaction, then this should be described to first order by a class in $H^2_a(\A,\A)$. This suggests that that class can be given by an element of $C^{1,1}(\A,\A)$.

In practice, computing a cohomology class means computing some cocycle in that class. This  requires making an additional choice. 

In principle, in order to compute the characteristic class of an interaction, we should (for each $M$) identify the algebras  for a family of field theories with a fixed vector space. It may then be possible to differentiate the product and homomorphisms to get a cocycle. Different choices of identifications should give cohomologous cocycles.

To understand what is needed, it is simplest to first consider classical field theories. A classical algebra of observables is a commutative algebra of functionals on the space of solutions. We need to choose a way of identifying solutions of  different field theories on a spacetime $M$. The simplest way to do this is by initial data. Given a Cauchy surface $\Sigma\subset M$, solutions of field theories on $M$ can be identified with their initial data on $\Sigma$. The characteristic class of an interaction can thus be computed by using an arbitrary choice of a Cauchy surface for every spacetime $M\in \Obj(\Loc)$.

To see how this can work, first consider a classical field theory given by some Lagrangian density, $\Lcal$, on a spacetime $M\in\Obj(\Loc)$. Let $\Sigma_1,\Sigma_2\subset M$ be Cauchy surfaces, and suppose for simplicity that\footnote{See Sec.~\ref{Notation} for notation.} $\Sigma_1\lesssim \Sigma_2$ and that these are equal outside of some compact set.

Assume for simplicity that this theory has no gauge degeneracies, so that the phase space is just the set of solutions of the equations of motion for $\Lcal$.
Such a solution can be identified with initial data along $\Sigma_1$ or $\Sigma_2$. If we change the Lagrangian density to $\Lcal+\lambda\Vcal$, then the equations of motion change, and evolution from $\Sigma_1$ to $\Sigma_2$ defines a map from the phase space to itself. Differentiating with respect to $\lambda$ and then setting $\lambda=0$ gives a vector field, $\Xi$, on the phase space.

Let
\[
H := \int_{\{x\in M\mid \Sigma_1\lesssim x \lesssim \Sigma_2\}} \Vcal \ .
\]
This is a functional on the phase space, and $\Xi$ is the Hamiltonian vector field given by $H$. (This is immediate from Peierls' construction of the Poisson structure \cite{pei1952,haa1996}.) 

If we define 
\[
\theta_i(x) := 
\begin{cases}
1 & x\gtrsim \Sigma_i \\
0 & \text{otherwise}
\end{cases}
\]
then 
\beq
\label{hamiltonian}
H = \int_M (\theta_1-\theta_2) \Vcal 
\eeq
(and we no longer need to assume $\Sigma_1\lesssim \Sigma_2$). 

This is good enough for classical field theory, but for a quantum field theory, $\Vcal$ will need to be a distribution, which is best thought of as a linear map,
\[
V:f \mapsto \int_Mf\Vcal 
\]
from $\Dt(M)$ to the algebra of observables. Unfortunately, $\theta_1-\theta_2$ is not smooth.

Note that because $\Sigma_i$ is a Cauchy surface, $\Supp \theta_i = J^+\Sigma_i$ is past compact and $\Supp (1-\theta_i) = J^-\Sigma_i$ is future compact. (See \cite{bar} for a discussion of these concepts.) This suggests  a smoothed-out analogue of  the set of Cauchy surfaces:
\begin{definition}
\label{Cauchy}
\[
\Theta(M) := \left\{\theta\in\Ci(M,\R)\mid \text{$\Supp\theta$ past compact, $\Supp(1-\theta)$ future compact}\right\} 
\]
\end{definition}

So, now let $\theta_1,\theta_2\in\Theta(M)$ and suppose that $\theta_1-\theta_2$ has compact support. Even if $\Vcal$ is distributional, we can define $H = V(\theta_1-\theta_2)$. The action of this on a classical observable, $a$, is the Poisson bracket,
\[
\Xi(a) =\{H,a\} = \left\{V(\theta_1-\theta_2),a\right\} \ .
\]

Note that if $\Supp f \sim \Supp a$, then $\left\{V(f),a\right\} = 0$, therefore
\beq
\label{H poisson}
\Xi(a) = \left\{V(\chi),a\right\} 
\eeq
for any test function $\chi\in\Dt(M)$ such that $\chi=\theta_1-\theta_2$ over $J(\Supp a)$.

This allows us to drop the assumption that $\theta_1-\theta_2$ has compact support. By construction, $\Supp(\theta_1-\theta_2)$ is future compact and past compact, so $\Supp(\theta_1-\theta_2) \cap J(\Supp a)$ is compact, and there exists $\chi\in\Dt(M)$ equal to $\theta_1-\theta_2$ on $J(\Supp a)$. Equation~\eqref{H poisson} can then be taken as the definition of $\Xi$.

This was for a classical field theory. For a quantum field theory, $\Xi$ should be a derivation of the algebra $\A(M)$ of quantum observables. In the classical limit, the commutator is approximately proportional to the Poisson bracket. This suggests that the quantum version of \eqref{H poisson} may be
\[
\Xi(a) = \tfrac{-i}{\hbar}[V(\chi),a] \ .
\]
This is indeed a derivation, and we shall see that it is the right answer.

\subsection{The character of an interaction}
Let $\Xc\subset \Loc$ be a small subcategory that is closed under pullbacks (i.e., intersections) and such that the inclusion of any causally complete open subset into $M\in\Obj(\Xc)$ is a morphism in $\Xc$. Let $\A$ be a functor from $\Xc$ to topological $\co$-algebras, satisfying Einstein causality.

In perturbative AQFT, the algebra $\A(M)$ is constructed from functionals, which have clearly defined support in $M$. However, I am trying to be more general here. Instead of defining the support of an observable, I have the following definitions for the subalgebra of observables supported on a given subset and for the subalgebra of compactly supported observables.
\begin{definition}
For any $M\in\Obj\Xc$ and $K\subset M$, let
\[
\A(M;K) := 
\bigcap\{\Im\A[\iota]\mid \forall \Or\in\Obj\Xc,\ \iota:\Or\to M,\text{ such that } K\subset \Im \iota \}
\]
and 
\[
\A_c(M):= \bigcup_{K\subset M\text{ compact}} \A(M;K)\ .
\]
For $\phi:M\to N$, $\A_c[\phi]$ is the restriction of $\A[\phi]$ to $\A_c(M)$.
\end{definition}
\begin{lem}
$\A_c:\Xc\to\Alg$ is a functor.
\end{lem}
\begin{proof}
Any $a\in\A_c(M)$ is in $\A(M;K)$ for some compact $K\subset M$. Consider some $\phi:M\to N$. For $\iota_1:\Or_1\to N$, let $\iota_2:\Or_2\to M$ be the pullback of $\iota_1$ by $\phi:M\to N$. If  $\phi(K)\subset\Im\iota_1$, then $K\subset\Im\iota_2$ so $a\in\Im\A[\iota_2]$ and 
\[
\A(\phi;a) \in \Im\A[\phi\circ\iota_2] \subseteq \Im\A[\iota_1] \ .
\]
Therefore,
\[
\A(\phi;a) \in \A(N;\phi(K)) \subset \A_c(N) \ .
\]
\end{proof}

Let $\Int:\Dt\naturalto\A$ be a linear natural transformation. The idea is to use this as an interaction term for a Lagrangian, $L +\lambda V$, although this  works even if $\A$ is not given by a Lagrangian, and it certainly doesn't need to be free. In perturbative AQFT, it is normally only assumed that $\Int$ is additive, rather than linear; however, we are only interested in the first order effect of this interaction, and only the linear part of $V$ will be relevant.

In order to choose a specific cocycle in the character of the interaction $\Int$, we need a choice of smoothed-out Cauchy surface on each spacetime, so following Definition~\ref{Cauchy} define:
\begin{definition}
\[
\Theta(\Xc) := \prod_{M\in\Obj\Xc}\Theta(M) \ .
\]
\end{definition}

\begin{lem}
For $V:\Dt\naturalto\A$, $M\in\Obj(\Xc)$, and $f\in\Dt(M)$,
\[
\Int_M(f) \in \A(M;\Supp f)
\]
\end{lem}
\begin{proof}
Consider any $\iota:\Or\to M$  such that $\Supp f\subset \Im\iota$, and note that $f=\iota_*\iota^*f$. By naturality of $\Int$,
\[
\Int_M(f) = \Int_M(\iota_*\iota^*f) = \A(\iota;\Int_\Or(\iota^*f)) \in \Im\A[\iota] \ . 
\]
Because this holds for any such $\iota$, the result follows.
\end{proof}

\begin{lem}
If $K_1\sim K_2\subset M$ are spacelike separated and compact, $a\in \A(M;K_1)$, and $b\in\A(M;K_2)$, then $ab=ba$.
\end{lem}
\begin{proof}
Let $K_1'\subset M$ be the causal complement of $K_1$. This is an open neighborhood of $K_2$. 

Because $K_1'$ is locally compact, every point of $K_2$ has a relatively compact (i.e., with compact closure) open neighborhood. This gives a cover of $K_2$ by relatively compact open subsets of $K_1'$. By compactness of $K_2$, this has a finite subcover. The union is a relatively compact open neighborhood of $K_2$. Let $\Or_2\subset K_1'\subset M$ be the Cauchy development of this neighborhood. Let $\Or_1:=\overline{\Or_2}'\subset M$ be the causal complement of $\overline{\Or_2}$.

Write $\iota_i:\Or_i\to M$ for the inclusions. Because $\Or_1$ and $\Or_2$ are globally hyperbolic neighborhoods of $K_1$ and $K_2$, $a\in\Im\A[\iota_1]$ and $b\in\Im\A[\iota_2]$. Because these are spacelike separated, $a$ and $b$ commute.
\end{proof}

\begin{lem}
\label{Cutoff_independent}
For $\Int:\Dt\naturalto\A$ additive, $K\subset M$ compact, $a\in\A(M;K)$, and  $f,g\in\Dt(M)$, if  $K\sim \Supp(f-g)$, then
\[
[\Int_M(f),a]=[\Int_M(g),a]\ ,
\]
where the bracket is the commutator.
\end{lem}
\begin{proof}
$J(K)$ is closed, so $J(K)\cap \Supp f$ is compact. There exists another function $h\in\Dt(M\smallsetminus \Supp(f-g))$ that equals $f$ (and hence $g$) over $J(K)\cap \Supp f$. 

By definition, $\Supp h$ is disjoint from $\Supp (f-g)$, so additivity of $\Int_M$ gives
\begin{align*}
\Int_M(h+[g-h]+[f-g]) &= \Int_M(h+[g-h]) - \Int_M(g-h) \\ & \qquad+ \Int_M([g-h]+[f-g]) \\
\Int_M(f) &= \Int_M(g) - \Int_M(g-h) + \Int_M(f-h)\ .
\end{align*}
In particular, $\Supp (g-h)$ and $\Supp (f-h)$ are spacelike to $K$, so
\[
[\Int_M(f),a]-[\Int_M(g),a] = [\Int_M(f-h),a]-[\Int_M(g-h),a] = 0\ .\qedhere
\]
\end{proof}

\begin{figure}
\begin{tikzpicture}[x=1.00mm, y=1.00mm, inner xsep=0pt, inner ysep=0pt, outer xsep=0pt, outer ysep=0pt]
\path[line width=0mm] (-2.00,-81.85) rectangle +(124.00,84.20);
\definecolor{F}{rgb}{0.753,0.753,0.753}
\path[fill=F] (28.00,-15.00) .. controls (57.19,-15.00) and (64.33,-15.00) .. (93.00,-15.00) .. controls (100.00,-15.00) and (91.04,-57.12) .. (85.71,-57.88) .. controls (65.88,-60.71) and (52.28,-60.57) .. (34.44,-57.46) .. controls (30.37,-56.75) and (20.00,-15.00) .. (28.00,-15.00) -- cycle;
\definecolor{L}{rgb}{0,0,0}
\path[line width=0.30mm, draw=L, dash pattern=on 2.00mm off 1.00mm] (0.00,-10.00);
\path[line width=0.30mm, draw=L, dash pattern=on 2.00mm off 1.00mm] (0.00,-15.00) -- (120.00,-15.00);
\path[line width=0.30mm, draw=L, dash pattern=on 2.00mm off 1.00mm] (0.00,-25.00) -- (120.00,-25.00);
\path[line width=0.30mm, draw=L] (60.00,-40.00) ellipse (5.00mm and 5.00mm);
\path[line width=0.30mm, draw=L, dash pattern=on 0.30mm off 0.50mm] (52.00,-40.00) -- (12.98,0.12);
\path[line width=0.30mm, draw=L, dash pattern=on 0.30mm off 0.50mm] (52.00,-40.00) -- (12.04,-79.62);
\path[line width=0.30mm, draw=L, dash pattern=on 0.30mm off 0.50mm] (68.39,-40.34) -- (107.91,0.35);
\path[line width=0.30mm, draw=L, dash pattern=on 0.30mm off 0.50mm] (68.63,-40.34) -- (107.44,-79.85);
\draw(58.57,-41.74) node[anchor=base west]{$K$};
\draw(70.00,-70.00) node[anchor=base west]{$M$};
\draw(90.14,-4.56) node[anchor=base west]{$J_N(K)$};
\draw(1.52,-12.74) node[anchor=base west]{$\theta_N=1$};
\draw(1.52,-30.28) node[anchor=base west]{$\theta_N=0$};
\draw(21.00,-33.00) node[anchor=base west]{$\theta_M=1$};
\draw(34.02,-64.65) node[anchor=base west]{$\theta_M=0$};
\path[line width=0.30mm, draw=L] (0.00,-46.00) .. controls (15.55,-29.41) and (37.27,-20.00) .. (60.00,-20.00) .. controls (82.73,-20.00) and (104.45,-29.41) .. (120.00,-46.00);
\path[line width=0.30mm, draw=L] (0.00,-46.00) .. controls (15.55,-62.59) and (37.27,-72.00) .. (60.00,-72.00) .. controls (82.73,-72.00) and (104.45,-62.59) .. (120.00,-46.00);
\path[line width=0.30mm, draw=L, dash pattern=on 2.00mm off 1.00mm] (0.00,-46.00) .. controls (18.47,-36.15) and (39.07,-31.00) .. (60.00,-31.00) .. controls (80.93,-31.00) and (101.53,-36.15) .. (120.00,-46.00);
\path[line width=0.30mm, draw=L, dash pattern=on 2.00mm off 1.00mm] (0.00,-46.00) .. controls (18.36,-55.06) and (38.53,-59.85) .. (59.00,-60.00) .. controls (80.15,-60.15) and (101.04,-55.36) .. (120.00,-46.00);
\end{tikzpicture}%
\caption{\label{Inclusion}An example of the possible arrangement of subsets of $N$ in Lemma~\ref{Character}. The shaded region represents $\Supp\chi$.}
\end{figure}
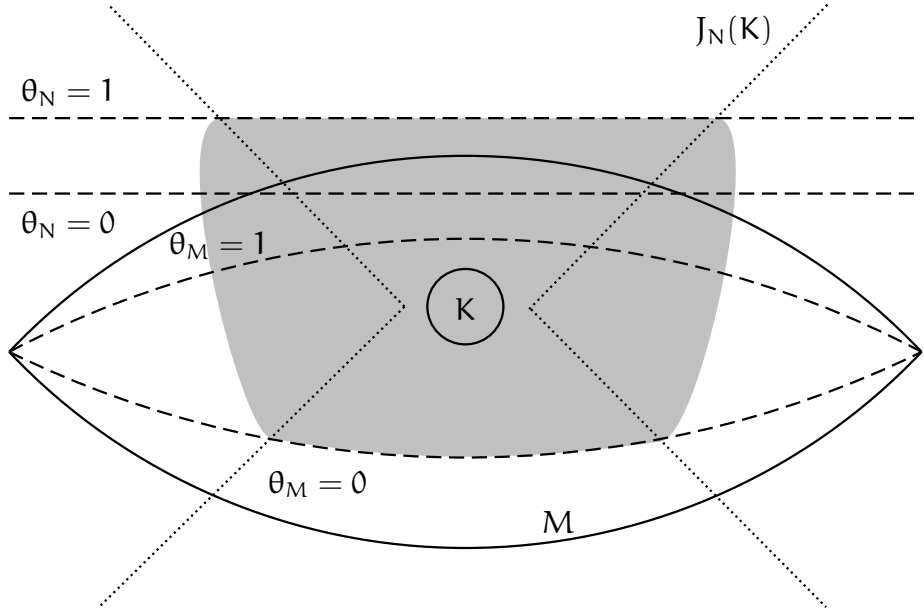
\begin{definition}
Given $\phi:M\to N$ and $\theta_M\in\Theta(M)$, define $\tilde \theta_M\in\Ci(J_N(M))$ by
\begin{align*}
\tilde\theta_M(\phi(x)) &= \theta_M(x) && x\in M\\
\tilde\theta_M(x)&= 0 && x\in J^-_N(M)\smallsetminus M\\
\tilde\theta_M(x)&= 1 && x\in J^+_N(M)\smallsetminus M\ .
\end{align*}
\end{definition}
To avoid clutter I am sometimes not writing $\phi$, as in $J^+_N(M)$ meaning $J^+_N(\phi M)$.
\begin{lem}
\label{Character}
For $\Int:\Dt\naturalto\A$ additive and any $\theta\in\Theta(\Xc)$,  there exists a unique $\Xi_{\Int,\theta}\in C^{1,1}(\A_c,\A_c)$ such that for $\phi:M\to N$, $K\subset M$ compact, and $a\in\A(M;K)$, if $\chi\in\Dt(N)$ such that $\chi=\tilde\theta_M-\theta_N$ on a neighborhood of $J_N(K)$ then $\Xi_{\Int,\theta}$ is given by the commutator
\[
\Xi_{\Int,\theta}(\phi;a) = \tfrac{-i}{\hbar}[\Int_N(\chi),\A(\phi;a)] \ .
\]
(See Fig.~\ref{Inclusion}.)\end{lem}

\begin{proof}
Let $\Or\subset M$ be a relatively compact neighborhood of $K$.
Observe that $\tilde\theta_M-\theta_N\in\Ci(J_N(M))$ has future and past compact support, so $J^\pm(\overline\Or)\cap\Supp(\tilde\theta_M-\theta_N)$ are compact, and therefore, $J(\overline\Or)\cap\Supp(\tilde\theta_M-\theta_N)$ is as well. Therefore there exists a function $\chi\in\Dt(N)$ that equals $\tilde\theta_M-\theta_N$ on $J_N(\overline\Or)\supset J_N(\Or)$. This satisfies the conditions in the hypothesis, so such functions do exist.

By the previous lemma, $\tfrac{-i}{\hbar}[\Int_N(\chi),\A(\phi;a)]$ is independent of the choice of such a function. This is clearly linear in $a\in\A(M;K)$, but any $a,b\in\A_c(M)$ are contained in $\A(M;K)$ for some $K$. Therefore this is a well defined linear map.
\end{proof}

\begin{lem}
\label{Closed}
If $\Int:\Dt\naturalto\A$ is linear and $\theta\in\Theta(\Xc)$, then $\delta\Xi_{\Int,\theta}=0$.
\end{lem}
\begin{proof}
Firstly, $\dH\Xi_{\Int,\theta}$ measures whether $\Xi_{\Int,\theta}[\phi]:\A_c(M)\to\A_c(N)$ is a derivation.  By construction, it is an inner derivation on $\A(M;K)$ for any compact $K\subset M$. Since $\A_c(M)$ is the union of these algebras, this shows that $\dH\Xi_{\Int,\theta}=0$. 

Now consider two composable morphisms $P\stackrel{\psi}{\longleftarrow} N \stackrel{\phi}{\longleftarrow} M$, some compact $K\subset M$, and $a\in\A(M;K)$. 
Choose a relatively compact neighborhood $\Or\subset M$ of $K$, and functions $\chi_\psi\in\Dt(P)$ and $\chi_\phi\in\Dt(N)$ such that $\chi_\psi = \tilde\theta_N-\theta_P$ on $J_P(\Or)$ and $\chi_\phi=\tilde\theta_M-\theta_N$ on $J_N(\Or)$.

Note that $\chi_{\psi\circ\phi} := \chi_\psi + \psi_*\chi_\phi$ satisfies
\begin{align*}
x\in J_M(\Or) &\implies & \chi_{\psi\circ\phi}(\psi\circ\phi(x)) &= \theta_N(\phi(x))  -\theta_P(\psi\circ\phi(x)) \\ &&&\qquad+\theta_M(x)- \theta_N(\phi(x))\\ &&&= \theta_M(x)-\theta_P(\psi\circ\phi(x)) 
\end{align*}
\begin{align*}
x \in J_N^-(\Or)\smallsetminus \Or &\implies & \chi_{\psi\circ\phi}(\psi(x)) &= \chi_\psi(\phi(x)) + \chi_\phi(x)\\ &&&= \theta_N(\phi(x)) - \theta_P(\psi\circ\phi(x)) - \theta_N(\phi(x)) \\ &&&= -\theta_P(\psi\circ\phi(x)) 
\end{align*}
\begin{align*}
x \in J_N^+(\Or)\smallsetminus \Or &\implies & \chi_{\psi\circ\phi}(\psi(x)) &= \theta_N(\phi(x))-\theta_P(\psi\circ\phi(x)) \\ &&&\qquad+1- \theta_N(\phi(x))   \\ &&&= 1-\theta_P(\psi\circ\phi(x))  
\end{align*}
\begin{align*}
x \in J_M^-(\Or)\smallsetminus N &\implies& \chi_{\psi\circ\phi}(x) &= \chi_\psi(x) = -\theta_P(x) 
\end{align*}
\begin{align*}
x \in J_M^+(\Or)\smallsetminus N &\implies &\chi_{\psi\circ\phi}(x) &= \chi_\psi(x) = 1-\theta_P(x)\ .
\end{align*}
So, $\chi_\psi$, $\chi_\phi$, and $\chi_{\psi\circ\phi}$ satisfy the hypotheses of Lemma~\ref{Character} for computing $\Xi_{\Int,\theta}(\psi;\A(\phi;a))$, $\Xi_{\Int,\theta}(\phi;a)$, and $\Xi_{\Int,\theta}(\psi\circ\phi;a)$.
Using the linearity and naturality of $\Int$, this gives
\[
\Int_P(\chi_{\psi\circ\phi}) = \Int_P(\chi_\psi) + \Int_P(\psi_*\chi_\phi) = \Int_P(\chi_\psi) +\A(\psi;\Int_N(\chi_\phi))\ .
\]
Therefore
\begin{multline*}
\dS\Xi_{\Int,\theta}(\psi,\phi;a) =\\ \A(\psi;\tfrac{-i}\hbar[\Int_N(\chi_\phi),\A(\phi;a)]) - \tfrac{-i}\hbar[\Int_P(\chi_{\psi\circ\phi}),\A(\psi\circ\phi;a)]+ \tfrac{-i}\hbar[\Int_P(\chi_{\psi}),\A(\psi\circ\phi;a)]  \\ = 0\ .
\end{multline*}
\end{proof}

\begin{lem}
\label{Cohomologous}
For $\Int:\Dt\naturalto\A$ linear and $\theta,\theta'\in\Theta(\Xc)$,
there exists $\Lambda_{\Int,\theta'-\theta}\in C^{0,1}(\A_c,\A_c)$ such that for any $M\in\Obj\Xc$, $K\subset M$ compact, $a\in\A(M;K)$, and $\xi\in\Ci(M)$ satisfying $\xi(x)=\theta'_M(x)-\theta_M(x)$ for $x\in J(K)$, we have
\[
\Lambda_{\Int,\theta'-\theta}(M;a) = \tfrac{-i}\hbar[\Int_M(\xi),a]\ ,
\]
and $\Xi_{\Int,\theta'} = \Xi_{\Int,\theta} + \delta\Lambda_{\Int,\theta'-\theta}$.
\end{lem}
\begin{proof}
Clearly, $\Supp (\theta'_M-\theta_M)$ is future and past compact, so $J(K)\cap\Supp(\theta'_M-\theta_M)$ is compact, and there exists a function $\xi\in\Dt(M)$ that equals $\theta'_M-\theta_M$ on this subset. By Lemma \ref{Cutoff_independent}, $\Lambda_{\Int,\theta'-\theta}(M;a)$ is independent of the choice of $\xi$, so $\Lambda_{\Int,\theta'-\theta}\in C^{0,1}(\A_c,\A_c)$ is well-defined.

By construction, $\Lambda_{\Int,\theta'-\theta}[M]$ is a derivation on any $\A(M;K)$, so $\dH\Lambda_{\Int,\theta'-\theta}=0$.

Now, let $\phi:M\to N$ and choose $\chi\in\Dt(N)$, $\xi_M\in\Dt(M)$, and $\xi_N\in\Dt(N)$ suitable to compute $\Xi_{\Int,\theta}(\phi;a)$, $\Lambda_{\Int,\theta'-\theta}(M;a)$, and $\Lambda_{\Int,\theta'-\theta}(N;\A(\phi;a))$, respectively. If we define $\chi':=\chi+\phi_*\xi_M-\xi_N$, then 
\begin{align*}
x\in J_M(K) &\implies& \chi'(\phi(x)) &= \theta_M(x)-\theta_N(\phi(x)) + \xi_M(x)-\xi_N(\phi(x))\\ &&&= \theta'_M(x)-\theta'_N(\phi(x))\\
x\in J^-_N(K)\smallsetminus M &\implies& \chi'(x) &= -\theta_N(x)-\xi_N(x) = -\theta'_N(x)\\
x\in J^+_N(K)\smallsetminus M &\implies& \chi'(x) &= 1-\theta_N(x)-\xi_N(x) = 1-\theta'_N(x)
\end{align*}
therefore (using linearity of $\Int$)
\begin{align*}
\Xi_{\Int,\theta'}(\phi;a) &= \tfrac{-i}{\hbar}[\Int_N(\chi'),\A(\phi;a)]\\
&= \tfrac{-i}{\hbar}[\Int_N(\chi),\A(\phi;a)] + \tfrac{-i}{\hbar}[\Int_N(\xi_N),\A(\phi;a)] - \tfrac{-i}{\hbar}[\Int_N(\phi_*\xi_M),\A(\phi;a)] \\
& = \Xi_{\Int,\theta}(\phi;a) + \A(\phi;\Lambda_{\Int,\theta'-\theta}(M;a))-\Lambda_{\Int,\theta'-\theta}(N;\A(\phi;a))\\
&= \Xi_{\Int,\theta}(\phi;a) +\dS\Lambda_{\Int,\theta'-\theta}(\phi;a)\ . \qedhere
\end{align*}
\end{proof}

Putting these results together proves:
\begin{thm}
\label{Characteristic class}
If $\Int:\Dt\naturalto\A$ is linear, then it defines a cohomology class $\Xi_\Int\in H^2_a(\A_c,\A_c)$, which is the class of $\Xi_{\Int,\theta}\in C^{1,1}(\A_c,\A_c)$ for any $\theta\in\Theta(\Xc)$.
\end{thm}

It is not completely clear whether linearity of $\Int$ is a necessary assumption, as only additivity was needed in Lemma~\ref{Character}. 

A test function $f\in\Dt(M)$ in $\Int_M(f)$ serves as an infrared cutoff of the interaction. It can be thought of as varying the coupling constant over $M$, and thus $\lambda V_M(f)$ should  really be $\Int_M(\lambda f)$. The characteristic class $\Xi_\Int$ is supposed to describe the first order effect of the interaction. From this perspective, even if $\Int$ is nonlinear, then only the first order, linear part of $\Int$ should be used to construct $\Xi_\Int$.

I have tried to prove the results in this section as generally as possible, but this has led to constructing $\Xi_V\in H^2_a(\A_c,\A_c)$ rather than in $H^2_a(\A,\A)$. In order to fix this, we need to use some more specific category of topological algebras such that $\A_c(M)\subset\A(M)$ is dense and $\Xi_{\Int,\theta}$ extends uniquely and continuously. This is true in the setting of perturbative AQFT.


\section{Perturbative AQFT}
\label{Perturbative}
Now turn to the setting of perturbative algebraic quantum field theory. Let $\hAlg$  denote the category of $\hbar$-adically complete $\co[[\hbar]]$-algebras and  homomorphisms.

The $\hbar$-adic topology can be defined by a norm such as $\norm{a}=e^{-k}$, where $\hbar^k$ is the largest power of $\hbar$ dividing $a$. Completeness of $A\in\Obj\hAlg$ means that any power series in $\hbar$ with coefficients in $A$ converges to an element of $A$. 

Any $\co[[\hbar]]$-linear map between $\hbar$-adically complete modules is norm-contracting, and hence continuous. If $A$ is a dense $\co[[\hbar]]$-submodule of a complete module and $\hbar a\in A \implies a\in A$, then any linear map from $A$ to a complete module extends uniquely to a homomorphism defined on the closure of $A$.

Similarly, $\hlAlg$ denotes the category of complete $\co[[\hbar,\lambda]]$-algebras.

\begin{remark}
These should really also be $*$-algebras, but for simplicity, I am ignoring the involution here.
\end{remark}

As in the previous section, $\Xc\subset\Loc$ is a small subcategory, closed under pullbacks and inclusion of causally complete open sets.
Let $\A:\Xc\to \hAlg$ satisfy Einstein causality and $\Int:\Dt\naturalto\A$ be additive.

\begin{definition}
A \emph{time-ordered product} $\dt$ on $\A(M)$ is a commutative, associative product such that if $\Or_1\gtrsim \Or_2\subset M$, $a\in\A(M;\Or_1)$, and $b\in\A(M;\Or_2)$, then $a\dt b = b\dt a = ab$.
\end{definition}

Suppose that we have made a natural choice of time-ordered product $\dt$ on each $\A(M)$; naturality, in this case, means that each $\A[\phi]$ is a homomorphism under the time-ordered products. Let $\ET$ denote the $\dt$ exponential function. 
 
 The formal S-matrix $\ess_M:\Dt(M)\to \A(M)[[\hbar^{-1}\lambda]]$ is $\ess_M(f):= \ET\{\frac{i}{\hbar}\Int_M(\lambda f)\}$; this satisfies the \emph{causal factorization property} that if $f,g,h\in\Dt(M)$ with $\Supp f \gtrsim \Supp h$, then
\beq
\label{causal factorization}
\ess_M(f+g+h) = \ess_M(f+g)\ess_M(g)^{-1}\ess_M(g+h) \ .
\eeq
Naturality of $V$ and $\dt$ imply that $\ess$ is natural.
\begin{definition}
 The \emph{retarded \Moller\ operator} 
\beq
\label{moller}
R_f : \A(M) \to \A(M)[[\lambda]]
\eeq
 is defined by $R_f(a) := \ess_M(f)^{-1}\left[\ess_M(f)\dt a\right]$. 
 \end{definition}
 This is a linear map, but not a homomorphism.
\begin{remark}
The formal S-matrix includes negative powers of $\hbar$, so it is slightly surprising that $R_f$ does not. See \cite{df2001,hr2016}. 
\end{remark}

 \begin{remark}
 I am not concerned with the details of renormalization here, but these are needed in order to actually construct $\dt$ and $\Int$.
 \end{remark}
 
 \begin{remark}
 The interaction $V$ is not actually used directly to construct the interacting theory. We only need a natural time-ordered product and a natural formal S-matrix satisfying eq.~\eqref{causal factorization} and \eqref{moller}
\end{remark}

Note that, for $a\in\A(M;K)$, 
\begin{align*}
\Supp f \gtrsim K &\implies  R_f(a)=a\ ,\\
\Supp f  \lesssim K & \implies  R_f(a) = \ess_M(f)^{-1}a\,\ess_M(f)\ .
\end{align*}
Roughly speaking, the interaction to the future of $K$ is irrelevant and the interaction to the past of $K$ only gives a unitary transformation. 

If we heuristically imagine that $f$ is $\{0,1\}$-valued, and $f=1$ on the causal completion $K''$ of $K$, then  $R_f(a)$ uses the interacting theory to evolve $a$ back to the past boundary of $\Supp f$, where it is identified with an observable of the free theory. This is good enough to identify interacting observables supported in $K$ with free observables. 

The problem is that, as we consider larger subsets of $M$, we must adjust $f$ and change the identification with free observables. No single choice of $f$ works for all observables, and this is why the algebraic adiabatic limit is needed.

\subsection{Adiabatic limit}
The following is a formalization of the algebraic adiabatic limit construction \cite{df2004} as applied to constructing an interacting LCQFT.

\begin{definition}
Define a strict functor $\K:\Xc\to\Cat$ (the category of small categories) by:
\begin{itemize}
\item 
For $M\in\Obj\Xc$, $\Obj \K(M)$ is the set of pairs $(K,f)$ where $K\subset M$ is compact and $f\in\Dt(M)$ such that $f=1$ on $J^+(K)\cap J^-(K)$. 
\item
For $M\in\Obj\Xc$, $\Mor \K(M)$ is the set of inclusions of compact subsets (without any condition relating the test functions).
\item
For $\phi:M\to N$, the  functor ($\Cat$ morphism) $\K[\phi]:\K(M)\to\K(N)$, is defined by $\K[\phi] : (K,f)\mapsto (\phi(K),\phi_*f)$.
\end{itemize}
\end{definition}

\begin{definition}
Given $M\in\Obj\Xc$, define a functor $\A_\Int(M;-):\K(M)\to\hlAlg$ such that:
\begin{itemize}
\item 
On objects, $\A_\Int(M;K,f)\subset \A(M)[[\lambda]]$ is generated by the image $R_f(\A(M;K))$. 
\item
On morphisms, it is determined by the condition that $R:\A(M;-)\naturalto \A_\Int(M;-)$ be a natural linear transformation --- i.e., for any $(K_1,f_1),(K_2,f_2)\in \K(M)$ with $K_1\subseteq K_2$, 
\[
\begin{tikzcd}
\A(M;K_1) \arrow{rrr}{\subseteq} \arrow{d}{R_{f_1}} &&& \A(M;K_2) \arrow{d}{R_{f_2}} \\
\A_\Int(M;K_1,f_1) \arrow{rrr}{\A_\Int(M;K_1,f_1,K_2,f_2)} &&& \A_\Int(M;K_2,f_2)
\end{tikzcd}
\]
\end{itemize}
is a commutative diagram.
\end{definition}

For any $\phi:M\to N$ and $(K,f)\in\K(M)$, we can consider the restriction of $\A[\phi]$ to $\A_\Int(M;K,f)\subset \A(M)[[\lambda]]$. To determine the codomain, observe that $\A[\phi]\circ R_f = R_{\phi_*f}\circ\A[\phi]$, so
\[
\A(\phi;\A_\Int(M;K,f)) \subseteq \A_\Int(N;\phi K,\phi_*f) = \A_\Int(M;\K[\phi](K,f)) \ .
\]
Now consider $(K_1,f_1,K_2,f_2)\in\Mor\K(M)$. Observe that, 
\begin{align*}
\A[\phi]\circ\A_\Int(M;K_1,f_1,K_2,f_2)\circ R_{f_1} &= \A[\phi]\circ R_{f_2} = R_{\phi_*f_2}\circ \A[\phi] \\
&= \A_\Int(N;\phi K_1,\phi_*f_1,\phi K_2,\phi_*f_2)\circ R_{\phi_*f_1}\circ\A[\phi] \\
&= \A_\Int(N;\phi K_1,\phi_*f_1,\phi K_2,\phi_*f_2) \circ\A[\phi] \circ R_{f_1} \ ,
\end{align*}
which implies that the restrictions of $\A[\phi]$ 
give a natural transformation 
\begin{equation}
\label{Al natural}
 \A_\Int(M;-) \naturalto  \A_\Int(N;-) \circ\K[\phi] \ .
\end{equation}

The transformation \eqref{Al natural} and universality of the direct limit $\varinjlim\A_\Int(M;-)$ give a homomorphism
\beq
\label{precomposition}
\varinjlim\A_\Int(M;-) \to \varinjlim\left\{\A_\Int(N;-)\circ\K[\phi]\right\} \ .
\eeq
Universality of $\varinjlim\left\{ \A_\Int(N;-)\circ\K[\phi]\right\}$ gives a homomorphism
\beq
\label{limit hom}
\varinjlim\left\{ \A_\Int(N;-)\circ\K[\phi]\right\} \to \varinjlim\A_\Int(N;-) \ .
\eeq
\begin{definition}
Define a functor $\A_\Int:\Xc\to\hlAlg$ by:
\begin{itemize}
\item
For $M\in\Obj\Xc$,
\[
\A_\Int(M) := \varinjlim \A_\Int(M;-)\ .
\]
\item
For $\phi:M\to N$, $\A_\Int[\phi] : \A_\Int(M)\to\A_\Int(N)$ is the composition of the homomorphisms \eqref{precomposition} and \eqref{limit hom}.
 \end{itemize}
 \end{definition}
This $\A_\Int$ is the perturbative AQFT given by modifying $\A$ with  the interaction $\lambda V$.

\subsection{Modified construction}
\begin{definition}
For $f,g\in\D(M)$ and $a\in\A(M)$, 
\beq
\label{mmoller}
\mm_{f,g}(a) := \ess_M(f)^{-1}[\ess_M(f+g)\dt a]\ess_M(f+g)^{-1}\ess_M(f) \ .
\eeq
\end{definition}
\begin{remark}
If we heuristically imagine that $f$ and $f+g$ are $\{0,1\}$-valued, then $\mm_{f,g}(a)$ evolves $a$ forward to the future boundary of $\Supp(f+g)$ and then back to the past boundary of $\Supp f$, where it is identified with an observable of the free theory.

This gives a uniform way of identifying interacting observables with free observables, because $\mm_{f,g}(a)$ does not change if $f$ is changed in the future and $g$ is changed in the past.
\end{remark}
\begin{figure}
\begin{tikzpicture}[x=1.00mm, y=1.00mm, inner xsep=0pt, inner ysep=0pt, outer xsep=0pt, outer ysep=0pt]
\definecolor{F}{rgb}{0.753,0.753,0.753}
\path[fill=F] (40,-28) arc (-175:-5:20mm) -- cycle;
\definecolor{L}{rgb}{0,0,0}
\path[line width=0.30mm, draw=L, dash pattern=on 2.00mm off 1.00mm] (1,-28) -- (119,-28);
\path[line width=0.30mm, draw=L, dash pattern=on 2.00mm off 1.00mm] (1,-53) -- (119,-53);
\path[line width=0.30mm, draw=L] (60.00,-40.00) [rotate around={360:(60.00,-40.00)}] ellipse (5.5mm and 5.5mm);
\path[line width=0.30mm, draw=L, dash pattern=on 0.30mm off 0.50mm] (56,-44) -- (22,-10);
\path[line width=0.30mm, draw=L, dash pattern=on 0.50mm off 0.50mm] (40,-28) -- (6,-62);
\path[line width=0.30mm, draw=L, dash pattern=on 0.30mm off 0.50mm] (64,-44) -- (98,-10);
\path[line width=0.30mm, draw=L, dash pattern=on 0.50mm off 0.50mm] (80,-28) -- (114,-62);
\draw(58.57,-41) node[anchor=base west]{$K$};
\draw(82,-15) node[anchor=base west]{$J^+(K)$};
\draw(82,-61) node[anchor=base west]{$J^-(K\cup\Supp g)$};
\draw(54.5,-26.5) node[anchor=base west]{$\theta=1$};
\draw(54.5,-57) node[anchor=base west]{$\theta=0$};
\draw(51.5,-33) node[anchor=base west]{$g=1-\theta$};
\draw(54.5,-50) node[anchor=base west]{$f=\theta$};
\end{tikzpicture}
\caption{\label{Modified Moller fig}In Theorem~\ref{Modified Moller}, $g=1-\theta$ in $J^+K$, and $f=\theta$ in $J^-(K\cup\Supp g)$. The shaded region is $\Supp g$.}
\end{figure}
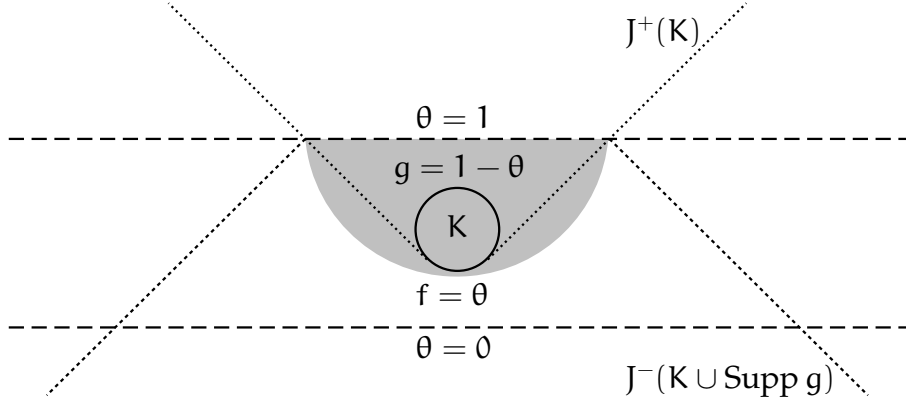
\begin{thm}
\label{Modified Moller}
Given $\theta\in\Theta(M)$, there exists a unique linear map (the modified \Moller\ operator) $\mm_\theta:\A(M)\to\A(M)[[\lambda]]$ such that if $K\subset M$ compact, $a\in\A(M;K)$, and $f,g\in\Dt(M)$ with $\Supp(g-1+\theta) \lesssim K$ and $\Supp(f-\theta)\gtrsim K\cup\Supp g$, then $\mm_\theta(a)=\mm_{f,g}(a)$.
\end{thm}
The different regions are sketched in Figure~\ref{Modified Moller fig}.
\begin{proof}
$\Supp(1-\theta)\cap J^+(K)$ is compact, so such a $g$ exists, and $(\Supp\theta) \cap J^-(K\cup\Supp g)$ is compact, so such an $f$ exists.

For a given choice of $g$, if $f'$ is another possible choice of $f$, then $\Supp(f-f')\gtrsim K\cup\Supp g$. This implies that
\[
\ess_M(f'+g)\dt a = \ess_M(f')\ess_M(f)^{-1}[\ess_M(f+g)\dt a] 
\]
and
\[
\ess_M(f'+g) = \ess_M(f')\ess_M(f)^{-1}\ess_M(f+g) \ ,
\]
so the right side of \eqref{mmoller} is independent of the choice of $f$. If $g'$ is another possible choice of $g$, then $\Supp\theta\cap J^-(K\cup\Supp g\cup\Supp g')$ is compact, so a choice of $f$ exists that is compatible with both $g$ and $g'$. Now, $\Supp (g-g') \lesssim K$, so
\[
\ess_M(f+g')\dt a = [\ess_M(f+g)\dt a]\ess_M(f+g)^{-1}\ess_M(f+g')\ ,
\]
and therefore the right side of \eqref{mmoller} is independent of the choice of $g$, i.e., $\mm_\theta(a)$ is well defined for $a\in \A(M;K)$.

For any $a,b\in\A_c(M)$, there exists $K\subset M$ compact such that $a,b\in \A(M;K)$. Any larger compact set determines the same $\mm_\theta(a)$, therefore $\mm_\theta$ is well defined. Because $a+b\in \A(M;K)$, we have $\mm_\theta(a+b)=\mm_\theta(a)+\mm_\theta(b)$, therefore $\mm_\theta$ is linear on $\A_c(M)$. Finally, this extends uniquely to $\A(M)$.
\end{proof}
\begin{figure}
\begin{tikzpicture}[x=1.00mm, y=1.00mm, inner xsep=0pt, inner ysep=0pt, outer xsep=0pt, outer ysep=0pt]
\definecolor{F}{rgb}{0.827,0.827,0.827}
\path[fill=F] (95.95,40.03) .. controls (103.47,32.84) and (109.73,12.84) .. (103.10,12.07) .. controls (84.53,9.93) and (55.35,9.35) .. (36.85,12.20) .. controls (30.21,13.22) and (35.88,32.49) .. (42.72,40.03) -- cycle;
\definecolor{F}{rgb}{0.65,0.65,0.65}
\path[fill=F] (54.21,29.39) .. controls (62.22,30.30) and (77.57,30.41) .. (85.89,29.36) .. controls (89.28,28.94) and (76.16,14.71) .. (70.16,14.71) .. controls (64.16,14.71) and (49.74,28.88) .. (54.21,29.39) -- cycle;
\definecolor{L}{rgb}{0,0,0}
\path[line width=0.30mm, draw=L, dash pattern=on 2.00mm off 1.00mm] (0.00,40.00) -- (140.00,40.00);
\path[line width=0.30mm, draw=L, dash pattern=on 2.00mm off 1.00mm] (0.00,0.00) -- (140.00,0.00);
\path[line width=0.30mm, draw=L, dash pattern=on 0.30mm off 0.50mm] (0.00,-2.00) -- (43.00,40.00);
\path[line width=0.30mm, draw=L, dash pattern=on 0.30mm off 0.50mm] (96.00,40.00) -- (140.00,-3.00);
\path[line width=0.30mm, draw=L, dash pattern=on 0.30mm off 0.50mm] (29.19,55.74) -- (67.00,17.00);
\path[line width=0.30mm, draw=L, dash pattern=on 0.30mm off 0.50mm] (111.90,56.12) -- (73.00,17.00);
\draw(68.50,18.50) node[anchor=base west]{$K$};
\path[line width=0.30mm, draw=L] (70.00,20.00) circle (4.24mm);
\draw(128.00,42.00) node[anchor=base west]{$\theta_N=1$};
\draw(125.00,-4.00) node[anchor=base west]{$\theta_N=0$};
\draw(115.00,10.00) node[anchor=base west]{$\theta_M=0$};
\draw(115.00,27.00) node[anchor=base west]{$\theta_M=1$};
\draw(59.5,26.50) node[anchor=base west]{$g_2=1-\theta_M$};
\path[line width=0.30mm, draw=L, dash pattern=on 0.30mm off 0.50mm] (53.95,29.39) -- (13.36,-9.74);
\path[line width=0.30mm, draw=L, dash pattern=on 0.30mm off 0.50mm] (86.37,29.00) -- (125.56,-10.00);
\draw(64.00,3.00) node[anchor=base west]{$f_2=\theta_N$};
\draw(87.91,-6.55) node[anchor=base west]{$M$};
\path[line width=0.30mm, draw=L] (0.00,20.00) .. controls (18.25,39.16) and (43.54,50.00) .. (70.00,50.00) .. controls (96.46,50.00) and (121.75,39.16) .. (140.00,20.00);
\path[line width=0.30mm, draw=L] (0.00,20.00) .. controls (18.25,0.84) and (43.54,-10.00) .. (70.00,-10.00) .. controls (96.46,-10.00) and (121.75,0.84) .. (140.00,20.00);
\path[line width=0.30mm, draw=L, dash pattern=on 2.00mm off 1.00mm] (0.00,20.00) .. controls (22.74,26.63) and (46.31,30.00) .. (70.00,30.00) .. controls (93.69,30.00) and (117.26,26.63) .. (140.00,20.00);
\path[line width=0.30mm, draw=L, dash pattern=on 2.00mm off 1.00mm] (140.00,20.00) .. controls  (117.26,13.37) and (93.69,10.00) .. (70.00,10.00) .. controls (46.31,10.00) and (22.74,13.37) .. (0.00,20.00);
\end{tikzpicture}%
\caption{\label{Modified functor}Some of the relationships in the proof of Theorem~\ref{Modified}. The darker shaded region is $\Supp g_2$. The lighter shaded region is $\Supp\chi$.}
\end{figure}
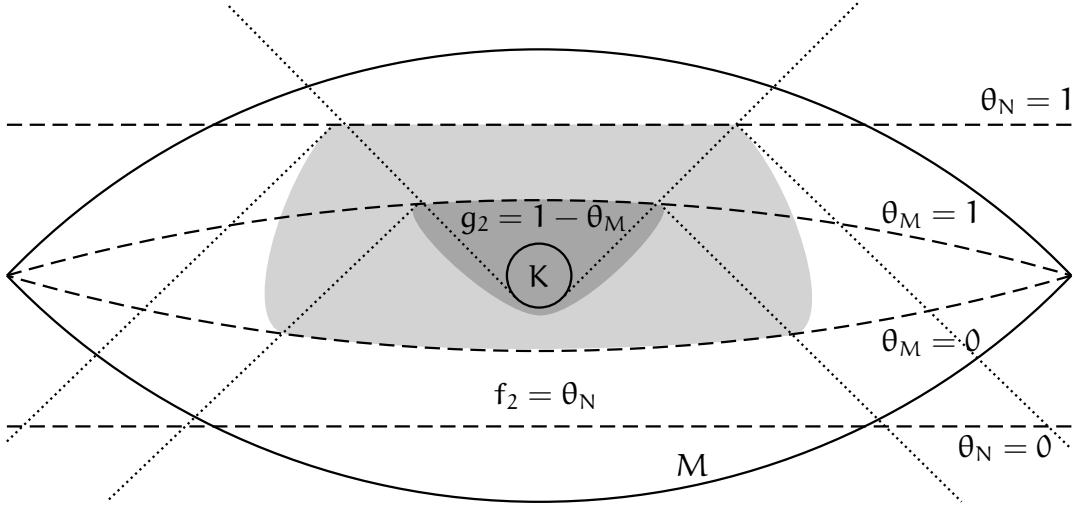
\begin{thm}
\label{Modified}
Given $\theta\in\Theta(\Xc)$, there is a unique functor $\A_{\Int,\theta}:\Xc\to\hlAlg$ such that: 
\begin{itemize}
\item
$\A_{\Int,\theta}(M)= \A(M)[[\lambda]]$;
\item
$\mm_\theta : \A\naturalto\A_{\Int,\theta}$ is a natural linear transformation.
\end{itemize}
\end{thm}
\begin{proof}
Because $\mm_{\theta_M}:\A(M)[[\lambda]]\to\A(M)[[\lambda]]$ is equal to the identity plus higher order terms in $\lambda$, it is automatically a bijective linear map.  
The naturality condition means that for any $\phi:M\to N$ 
\[
\begin{tikzcd}
\A(M) \arrow{d}{\mm_{\theta_M}} \arrow{r}{\A[\phi]} & \A(N) \arrow{d}{\mm_{\theta_N}} \\
\A(M)[[\lambda]] \arrow{r}{\A_{\Int,\theta}[\phi]} & \A(N)[[\lambda]]
\end{tikzcd}
\]
should commute.
Because $\mm_{\theta_M}$ is injective, this clearly defines a linear map from $\Im\mm_{\theta_M}$ to $\Im\mm_{\theta_N}$, and hence $\A(M)[[\lambda]] \to \A(N)[[\lambda]]$. It is automatically functorial. We need to check that it is a homomorphism.

Consider $a\in\A(M;K)$ and compare $\mm_{\theta_M}(a)$ with $\mm_{\theta_N}(\A(\phi;a))$. First, note that if $\mm_{\theta_M}(a)$ is computed with $f,g\in\Dt(M)$ as before, then
\[
\A(\phi;\mm_{\theta_M}(a)) = \ess_N(\phi_*f)^{-1}[\ess_N(\phi_*f+\phi_*g)\dt\A(\phi;a)]\ess_N(\phi_*f+\phi_*g)^{-1}\ess_N(\phi_*f)\ .
\]
In this formula, $\phi_*f$  can safely be replaced with a test function on $N$, as long as the support of the difference is $\gtrsim \phi(K\cup\Supp g)$.

With this in mind, choose $g_1\in\Dt(M)$ such that $\Supp(g_1-1+\theta_M)\lesssim K$. Again, denote by $\tilde\theta_M\in\Ci(J_N(M))$ the function such that $\phi^*\tilde\theta_M=\theta_M$ and $\tilde\theta_M=1$ in the future of $\phi(M)$ and $0$ in the past. Next choose $\chi\in\Dt(N)$ such that $\chi=\tilde\theta_M-\theta_N$ on $J(\phi(K))\cup J^-(\phi(\Supp g_1))$. Define $g_2=\phi_*g_1+\chi$. Choose $f_2\in\Dt(N)$ such that $\Supp(f_2-\theta_N)\gtrsim \phi(K\cup\Supp g_1)\cup\Supp g_2$. Finally, let $f_1=f_2+\chi$. (Some of this is sketched in Fig.~\ref{Modified functor}.)

The point of these choices is that they can be used to compute $\A(\phi;\mm_{\theta_M}(a))$ and $\mm_{\theta_N}(\A(\phi;a))$, and $f_1+\phi_*g_1=f_2+g_2$. Now,
\[
\A(\phi;\mm_{\theta_M}(a)) = \ess_N(f_1)^{-1}[\ess_N(f_2+g_2)\dt\A(\phi;a)]\ess_N(f_2+g_2)^{-1}\ess_N(f_1)
\]
and
\begin{align*}
\mm_{\theta_N}(\A(\phi;a)) &= \ess_N(f_2)^{-1}[\ess_N(f_2+g_2)\dt\A(\phi;a)]\ess_N(f_2+g_2)^{-1}\ess_N(f_2)\\
&= \ess_N(f_2)^{-1}\ess_N(f_1) \A(\phi;\mm_{\theta_M}(a)) \ess_N(f_1)^{-1} \ess_N(f_2)\ .
\end{align*}
This shows that for $b\in \mm_{\theta_M}(\A(M;K))$, 
\beq
\label{mm hom}
\A_{\Int,\theta}(\phi;b) = \ess_N(f_2)^{-1}\ess_N(f_1) \A(\phi;b) \ess_N(f_1)^{-1} \ess_N(f_2) \ .
\eeq
This $\A_{\Int,\theta}[\phi]$ is manifestly a homomorphism.

Finally, this extends uniquely to all of $\A_{\Int,\theta}(M)$.
\end{proof}
\begin{cor}
For any $\theta,\theta'\in\Theta(\Xc)$, there is a natural isomorphism  $\alpha:\A_{\Int,\theta}\naturalto\A_{\Int,\theta'}$ such that 
\[
\begin{tikzcd}
&\A(M) \arrow{dl}[swap]{\mm_{\theta_M}} \arrow{dr}{\mm_{\theta'_M}} \\
\A_{V,\theta}(M) \arrow{rr}{\alpha_M} && \A_{V,\theta'}(M)
\end{tikzcd}
\]
commutes for any $M\in\Obj\Xc$.
\end{cor}
\begin{proof}
The calculation is the same as in the previous proof, although slightly simplified. More formally, this result follows if we apply the previous theorem to a doubled version of the category $\Xc$ in which every object of $\Xc$ appears as 2 isomorphic copies.
\end{proof}
\begin{thm}
For any $\theta\in\Theta(\Xc)$, $\A_{\Int,\theta}$ is naturally isomorphic to $\A_\Int$.
\end{thm}
\begin{proof}
First, for $M\in\Obj\Xc$, $K\subset M$ compact, $\theta'\in\Theta(M)$ such that $\Supp(1-\theta')\lesssim K$, and $a\in\A(M;K)$, we can compute $\mm_{\theta'}(a)$ using $g=0$. So, if $\Supp(f-\theta')\gtrsim K$, then 
\[
R_f(a) = \mm_{f,0}(a)=\mm_{\theta'}(a)\ .
\]

Now, for any $(K,f)\in\K(M)$, there exists $\theta'\in\Theta(M)$ such that $\Supp(f-\theta')\gtrsim K$. So, the isomorphism that intertwines $\mm_{\theta'}$ with $\mm_{\theta_M}$ also maps $\A_\Int(M;K,f)$ to $\A_{\Int,\theta}(M)$. These maps are consistent with $\A_\Int(M;-)$, so these give a homomorphism from $\A_\Int(M) = \varinjlim \A_\Int(M;-)$ to $\A_{\Int,\theta}(M)$. For the same reason, this homomorphism is natural.

The image of $\A_\Int(M;K,f)$ is densely generated by $\mm_{\theta_M}(\A(M;K))$, but these subalgebras densely generate $\A_{\Int,\theta}(M)$, therefore this homomorphism is surjective. It is injective by construction, therefore it is an isomorphism.
\end{proof}

This shows that $\A_{\Int,\theta}$ is completely equivalent to $\A_\Int$. It has the advantage of being more concrete. Whereas $\A_\Int(M)$ is defined abstractly as a limit, $\A_{\Int,\theta}(M)$ is just $\A(M)[[\lambda]]$.

\subsection{Maurer-Cartan}
The computations in the previous section make it possible to explicitly show that the characteristic class $\Xi_V$ of an interaction is a Maurer-Cartan element of $H^2_a(\A,\A)$, provided that a natural time-ordered product exists. In this section $\A:\Xc\to\hAlg$ satisfies Einstein causality and has a natural time-ordered product as in the previous section.

If $V$ is nonlinear, then only its linear part is relevant at first order in $\lambda$, so in this section I will assume that $V:\Dt\naturalto\A$ is a natural linear transformation.

For $\theta\in\Theta(M)$, the map $\mm_{\theta}$ is a formal power series in $\lambda$, so denote this expansion explicitly as
\beq
\label{mm expansion}
\mm_\theta = \id + \lambda\mm_\theta^{(1)} + \tfrac12\lambda^2\mm_\theta^{(2)}+\dots\ .
\eeq
Inverting this gives
\beq
\label{mm inverse}
\mm_\theta^{-1} = \id - \lambda\mm^{(1)}_\theta + \lambda^2\left(\mm^{(1)}_\theta\circ\mm^{(1)}_\theta - \tfrac12 \mm^{(2)}_\theta\right) + \dots \ .
\eeq

\begin{thm}
For $\theta\in\Theta(\Xc)$, and  $\phi:M\to N$ in $\Xc$, 
\[
\A_{\Int,\theta}[\phi] \equiv \A[\phi]+ \lambda\, \Xi_{\Int,\theta}[\phi] \mod \lambda^2 \ ,
\]
where $\Xi_{V,\theta}\in C^{1,1}(\A,\A)$ is now $\hbar$-adically completed.
\end{thm}
\begin{proof}
This follows from eq.~\eqref{mm hom} by expanding $\ess_N$ to first order in $\lambda$.
\end{proof}

For $\theta\in\Theta(\Xc)$, I  would like to compute a cochain $\Psi_{V,\theta}$ such that $\A_{\Int,\theta} = \A+\lambda\Xi_{\Int,\theta} + \frac12\lambda^2\Psi_{V,\theta} +\dots$. By definition, for $\phi:M\to N$, 
\[
\A_{\Int,\theta}[\phi] = \mm_{\theta_N}\circ\A[\phi]\circ\mm_{\theta_M}^{-1}\ .
\]
Using \eqref{mm expansion} and \eqref{mm inverse}, this gives
\beq
\label{mm Xi}
\Xi_{\Int,\theta}[\phi] = \mm^{(1)}_{\theta_N}\circ\A[\phi] - \A[\phi]\circ\mm^{(1)}_{\theta_M} 
\eeq
and leads to:
\begin{definition}
For $V:\Dt\naturalto\A$ linear, $\theta\in\Theta(\Xc)$, let $\Psi_{V,\theta}\in C^{1,1}(\A,\A)$ be such that for $\phi:M\to N$,
\beq
\label{Psi def}
\Psi_{V,\theta}[\phi] = \mm_{\theta_N}^{(2)}\circ\A[\phi] - \A[\phi]\circ\mm_{\theta_M}^{(2)} - 2\,\Xi_{\Int,\theta}[\phi]\circ \mm^{(1)}_{\theta_M} \ .
\eeq
\end{definition}

\begin{lem}
\label{Psi explicit lemma}
For $V:\Dt\naturalto\A$ linear, $\theta\in\Theta(\Xc)$,  $\phi:M\to N$, $K\subset M$ compact, if $\Or\subset M$ is a neighborhood of
\[
K \cup (J^+K\cap\Supp [1-\theta_M])\cup (J^-K\cap\Supp \theta_M)\ ,
\]
and $\chi,f_2\in\D(N)$,  such that 
\begin{align*}
\chi &= \tilde\theta_M-\theta_N && \text{on $J_N \Or$}\\
f_2 &= \theta_N && \text{on $J^-_N(\Or\cup \Supp \chi)$}
\end{align*}
(where $\tilde\theta_M\in\Theta(J_N\Or)$ is again the extension of $\theta_M$ by $1$ in the future and $0$ in the past) then for any $a\in\A(M;K)$,
\begin{multline}
\label{Psi formula}
\Psi_{V,\theta}(\phi;a) = \tfrac{-1}{\hbar^2}[\Int_N(\chi),[\Int_N(\chi),\A(\phi;a)]] \\ + \tfrac{-1}{2\hbar^2}[\Int_N(2f_2+\chi)\dt\Int_N(\chi)-\Int_N(2f_2+\chi) \Int_N(\chi),\A(\phi;a)] \ .
\end{multline}
Moreover, for any such $\phi$ and $K$, such $\Or$, $\chi$, and $f_2$ exist.
\end{lem}
\begin{proof}
The map defined in \eqref{mmoller} is natural in the sense that for $\phi:M\to N$,
\[
\A[\phi]\circ \mm_{f,g} = \mm_{\phi_*f,\phi_*g}\circ\A[\phi] \ .
\]

Explicitly,
\[
\mm^{(1)}_{f,g}(a) = \tfrac{-i}{\hbar}\left(\Int_M(f)a+a \Int_M(g)-\Int_M(f+g)\dt a\right) \ .
\]
This only depends upon $f$ in $J^-K$ and $g$ in $J^+K$, therefore $\mm^{(1)}_{f,g}(a) = \mm^{(1)}_{\theta_M}(a)$ if $f=\theta_M$ in $J^-K$ and $g=1-\theta_M$ in $J^+K$. If $\Or\subset M$ is a neighborhood of 
\[
K \cup (J^+K\cap\Supp [1-\theta_M])\cup (J^-K\cap\Supp \theta_M) \ ,
\] 
then $f$ and $g$ can be chosen that are supported in $\Or$, therefore $\mm^{(1)}_{\theta_M}(a)$ is in the image of $\A[\iota]$, where $\iota:\Or\to M$ is the inclusion.

From the properties of $\mm$, $\mm^{(2)}_{f_2,g_2} \circ \A(\phi;a) = \mm^{(2)}_{\theta_N} \circ \A(\phi;a)$, provided that $g_2=1-\theta_N$ on $J^+_NK$, and $f_2=\theta_N$ on $J^-_N(K\cup\Supp g_2)$.

Likewise, for $f_1\in\D(N)$ and $g_1\in\D(M)$,
\[
\A[\phi]\circ\mm^{(2)}_{\theta_M}(a) = \mm^{(2)}_{f_1,\phi_*g_1}\circ\A(\phi;a)
\]
if $g_1=1-\theta_M$ on $J^+K$ and $f_1=\tilde\theta_M$ on $J^-_N(K\cup\Supp g_1)$.

Thirdly, for $\chi,f_1\in\D(N)$ and $g_1\in\D(M)$,
\begin{align*}
\Xi_{\Int,\theta}[\phi]\circ\mm^{(1)}_{\theta_M}(a) &= \tfrac{-i}{\hbar}[\Int_N(\chi),\A[\phi]\circ\mm^{(1)}_{\theta_M}(a)] \\
&= \tfrac{-i}{\hbar}[\Int_N(\chi),\mm^{(1)}_{f_1,\phi_*g_1}\circ\A(\phi;a)]
\end{align*}
if $\chi=\tilde\theta_M-\theta_N$ on $J_N\Or$, $f_1=\tilde\theta_M$ on $J^-_NK$, and $g_1=1-\theta_M$ on $J^+K$.

If all of these conditions are satisfied, then
\beq
\label{Psi explicit}
\Psi_{V,\theta}(\phi;a) = \left(\mm^{(2)}_{f_2,g_2}-\mm^{(2)}_{f_1,\phi_*g_1}\right)\circ \A(\phi;a) + \tfrac{2i}{\hbar}[\Int_N(\chi),\mm^{(1)}_{f_1,\phi_*g_1}\circ\A(\phi;a)] \ .
\eeq
If $f_1-f_2=g_2-\phi_*g_1=\chi$, then the right side of \eqref{Psi explicit} equals
\[
\tfrac{-1}{\hbar^2}[\Int_N(\chi),[\Int_N(\chi),\A(\phi;a)]] + \tfrac{-1}{2\hbar^2}[\Int_N(f_1+f_2)\dt\Int_N(\chi)-\Int_N(f_1+f_2) \Int_N(\chi),\A(\phi;a)] \ .
\]
This is most easily seen by expanding eq.~\eqref{mm hom} to second order in $\lambda$. Note that $g_1$ and $g_2$ do not appear in this expression.

Now suppose that $\Or\subset M$, and $\chi,f_2\in\D(N)$ satisfy the hypotheses of this lemma. $\Or\subset M$ is an open neighborhood of the compact set $J^+K\cap\Supp (1-\theta_M)$, therefore there exists a function in $\D(\Or)$ that equals $1$ on $J^+K\cap\Supp (1-\theta_M)$; multiplying this with $1-\theta_M$ gives a function $g_1\in\D(M)$ such that $g_1=1-\theta_M$ on $J^-K$ and $\Supp g_1\subset \Or$. Define $f_1:=f_2+\chi$ and $g_2:=\phi_*g_1+\chi$. The conditions are satisfied as follows:
\begin{itemize}
\item
$f_1=\tilde\theta_M$ on $J^-(K\cup\Supp g_1)$, because $K\subset\Or$, $\Supp g_1\subset\Or$, and $f_1=f_2+\chi=\tilde\theta_M$ on $J^-_N\Or$.
\item
$g_1=1-\theta_M$ on $J^+K$, by construction.
\item
$f_2=\theta_N$ on $J^-_N(K\cup\Supp g_2)$, because $\Supp g_2 \subset \Supp g_1\cup\Supp\chi\subset \Or\cup\Supp \chi$.
\item
$g_2=1-\theta_N$ on $J^+_NK$, because $\phi_*g_1+\chi = 1-\tilde\theta_M+(\tilde\theta_M-\theta_N)$ on $J^+_NK$.
\item
$\chi=\tilde\theta_M-\theta_N$ on $J_NK$, by hypothesis.
\end{itemize}
This verifies eq.~\eqref{Psi formula} under the given hypotheses.

Finally, consider any $\phi:M\to N$ and $K\subset M$ compact. The sets $J^+K\cap\Supp(1-\theta_M)$ and $J^-K\cap\Supp\theta_M$ are compact, because $\Supp(1-\theta_M)$ and $\Supp\theta_M$ are (respectively) future-compact and past-compact. Therefore,
\[
K_2:=K \cup (J^+K\cap\Supp [1-\theta_M])\cup (J^-K\cap\Supp \theta_M)
\]
is compact. Because $M$ is locally compact, every point has a precompact neighborhood. By compactness, $K_2$ has a finite cover by precompact open sets, therefore $K_2$ has a precompact open neighborhood, $\Or\subset M$.

The function $\tilde\theta_M-\theta_N$ is defined over $J_NM$, where it has future and past-compact support. Because the closure $\overline\Or$ is compact, there exists $\chi\in\D(N)$ with compact support in $J_NM$ and $\chi=\tilde\theta_M-\theta_N$ over $J\overline\Or$. 

Finally, because $\Supp\chi$ and $\overline\Or$ are compact and $\Supp\theta_N$ is past-compact, there exists $f_2\in\D(N)$ such that $f_2=\theta_N$ over $J^-_N(\overline\Or\cup\Supp\chi)$.
\end{proof}
\begin{remark}
I began by assuming the existence of a natural time-ordered product, $\dt$. However, in the construction of $\Psi_{V,\theta}$ it is only used on the image of $V$, not on arbitrary elements of the algebra. This suggests that it may be sufficient to only define $\dt$ on this subspace. 
\end{remark}

\begin{thm}
$0 = \delta\Psi_{V,\theta} + [\Xi_{\Int,\theta},\Xi_{\Int,\theta}]$.
\end{thm}
\begin{proof}
Both terms have components in degrees $(2,1)$ and $(1,2)$, so this is really two equations.

In degree $(2,1)$, we need to show that $0=\dS\Psi_{V,\theta} + 2\,\Xi_{\Int,\theta}\circ\Xi_{\Int,\theta}$. More explicitly, for any $P\stackrel{\psi}{\longleftarrow} N \stackrel{\phi}{\longleftarrow} M$, we need to show that $(\dS\Psi_{V,\theta})[\psi,\phi]=-2\,\Xi_{\Int,\theta}[\psi]\circ\Xi_{\Int,\theta}[\phi]$.

The first 2 terms in eq.~\eqref{Psi def} cancel perfectly in $\dS\Psi_{V,\theta}$. This leaves
\begin{align*}
(\dS\Psi_{V,\theta})[\psi,\phi] &= \Psi_{V,\theta}[\psi]\circ\A[\phi]-\Psi_{V,\theta}[\psi\circ\phi]+\A[\psi]\circ\Psi_{V,\theta}[\phi] \\
&= -2\,\Xi_{\Int,\theta}[\psi]\circ\mm^{(1)}_{\theta_N}\circ\A[\phi]+2\left(\Xi_{\Int,\theta}[\psi\circ\phi]-\A[\psi]\circ\Xi_{\Int,\theta}[\phi]\right)\circ\mm^{(1)}_{\theta_M} \ .
\end{align*}
Because $0=\dS\Xi_{\Int,\theta}$,
\begin{align*}
(\dS\Psi_{V,\theta})[\psi,\phi] &= -2\,\Xi_{\Int,\theta}[\psi]\circ\mm^{(1)}_{\theta_N}\circ\A[\phi]+ 2\,\Xi_{\Int,\theta}[\psi]\circ\A[\phi]\circ\mm^{(1)}_{\theta_M}\\
&= -2\,\Xi_{\Int,\theta}[\psi]\circ \Xi[\phi]
\end{align*}
by eq.~\eqref{mm Xi}.

In degree $(1,2)$, we need to show that $0=\dH\Psi_{V,\theta}+2\,\Xi_{\Int,\theta}\bullet\Xi_{\Int,\theta}$. More explicitly, for any $\phi:M\to N$ and $a,b\in\A(M)$, we need to show that $(\dH\Psi_{V,\theta})(\phi;a,b) = 2\,\Xi_{\Int,\theta}(\phi;a)\Xi_{\Int,\theta}(\phi;b)$.

Suppose that $a,b\in\A(M;K)$, for some compact $K\subset M$, and let $\chi$ be as in Lemma~\ref{Psi explicit lemma}. The last term of eq.~\eqref{Psi formula} is manifestly a derivation, so it doesn't contribute to $\dH\Psi_{V,\theta}$. Note that in any associative algebra, the commutator satisfies
\begin{align*}
[A,[A,BC]] &= [A,[A,B]C+B[A,C]] \\
&= [A,[A,B]]C + 2[A,B][A,C] + B[A,[A,C]] \ .
\end{align*}
This leaves
\begin{align*}
(\dH\Psi_{V,\theta})(\phi;a,b) &= -\A(\phi;a)\Psi_{V,\theta}(\phi;b)+\Psi_{V,\theta}(\phi;ab)-\Psi_{V,\theta}(\phi;a)\A(\phi;b) \\
&= \tfrac{-2}{\hbar^2}[\Int_N(\chi),\A(\phi;a)][\Int_N(\chi),\A(\phi;b)]\\
&= 2\, \Xi_{\Int,\theta}(\phi;a)\Xi_{\Int,\theta}(\phi;b) \ . \qedhere
\end{align*}
\end{proof}
In terms of the cohomology class $\Xi_V\in H_a^2(\A,\A)$ of $\Xi_{V,\theta}$, this means simply that it satisfies the Maurer-Cartan equation,
\beq
\label{MC eq}
[\Xi_V,\Xi_V]=0 \ .
\eeq
This is the analogue of the Jacobi identity satisfied by a Poisson structure.

\section{Conclusions}
This paper has presented several new ideas and perspectives on constructing models in algebraic quantum field theory.

The first is that the algebraic structures and properties of a locally covariant quantum field theory are naturally organized in the Hochschild bicomplex. This makes it possible to consider deformations of a quantum field theory as a generalization of  deformation quantization of a Poisson manifold. 

This perspective leads to the definition of \emph{skew diagrams of algebras} (Def.~\ref{Skew}) as a generalization of functors $\Xc\to\Alg$. Because the other axioms of AQFT are still meaningful for a skew diagram, this is a more general framework for building physical models.

An AQFT model is a functor $\A:\Xc\to\Alg$, and the global symmetries of the model (e.g., the $\U(1)$-symmetry associated to charge conservation) are described as the natural automorphisms of $\A$. This cohomological perspective leads to a more general definition of symmetry as \emph{skew automorphisms} (Def.~\ref{Skew morphism}) or \emph{outer skew automorphisms} (Def.~\ref{Outer}) of $\A$. This means that some models may have symmetries that were not previously recognized.

From this perspective, an interaction is analogous to a bivector field on a manifold. Just as bivectors can be constructed by multiplying and adding vectors, interactions may be constructed by multiplying and adding infinitesimal skew automorphisms, using the Gerstenhaber algebra structure on Hochschild cohomology. 
For an interaction of this form, the full interacting model may then be constructed directly by methods similar to those in \cite{bls2011}, generalizing the strict deformation quantization construction in \cite{rie11}.

One way of deforming AQFT is to begin with a model defined by a Lagrangian and then perturb it by adding an interaction term, $V$, to that Lagrangian. I have explicitly constructed (Thm.~\ref{Characteristic class}) the characteristic Hochschild cohomology class $\Xi_V$ of such an interaction term. However, this construction does not require the initial theory to be defined by a Lagrangian. This makes it possible to perturb a non-Lagrangian model with a Lagrangian interaction.

The construction of this characteristic class required introducing a notion (Def.~\ref{Cauchy}) of a smoothed-out Cauchy surface as a function that vanishes in the distant past and equals $1$ in the distant future.
In order to compare this characteristic class with the construction of perturbative AQFT, I introduced an alternative version (Thm.~\ref{Modified}) of that construction, which uses smoothed-out Cauchy surfaces instead of the algebraic adiabatic limit. This is more concrete than the standard construction, so this is likely to be advantageous for many calculations.

Finally, assuming the existence of a time-ordered product, I showed that the characteristic class of an interaction satisfies the appropriate Maurer-Cartan equation \eqref{MC eq}.

Much remains to be done. In particular, nontrivial examples of skew automorphisms would be very useful, because they can be used to construct interacting models. This cohomology is defined in a purely algebraic setting, so it needs to be extended or adapted to apply to \cs-algebras or von~Neumann algebras. 

This is not the first cohomological structure to be associated with quantum field theory. Rejzner used a BV-bicomplex to construct gauge theories in perturbative AQFT \cite{rej2016}. Hochschild cohomology needs to be extended to gauge theories, and combining the Hochschild complex with the BV-bicomplex may lead to further new perspectives.

\subsection*{Acknowledgments}
I wish to thank Chris Fewster and Kasia Rejzner for many discussions regarding this project. In particular, Kasia guided me through perturbative AQFT but is not responsible for any misunderstandings that I still cling to. Further conversations with Uli Kr\"{a}hmer, Urs Schreiber, and Birgit Richter have greatly improved this paper from the first version.

\end{document}